\documentclass[final,seceqn,amsthm,a4paper]{elsart}
\usepackage{graphicx,amsmath,amssymb,amsfonts,pstricks,pst-node,pst-plot,hyperref,algorithmic}
\journal{Nuclear Physics B}
\usepackage[plain]{algorithm}

\theoremstyle{plain}
\newtheorem{theorem}{Theorem}[section]
\newtheorem{proposition}[theorem]{Proposition}
\newtheorem{lemma}[theorem]{Lemma}

\newtheorem{algthm}{Algorithm}

\theoremstyle{definition}

\theoremstyle{remark}
\newtheorem{remark}{Remark}[section]

\begin{document}
\newcommand{\ontop}[2]{\genfrac{}{}{0pt}{}{#1}{#2}}
\newcommand{\pworm}{P_{w}}
\newcommand{\pwormbar}{\overline{P}_{w}}
\newcommand{\piworm}{\pi_{w}}
\newcommand{\piwormbar}{\overline{\pi}_{w}}
\newcommand{\pimodworm}{\pi_{w}'}
\newcommand{\pmodworm}{P_{w}'}
\newcommand{\pimodwormbar}{\overline{\pi'}_{w}}
\newcommand{\pmodwormbar}{\overline{P'}_{w}}
\newcommand{\lhs}{\text{LHS}}
\newcommand{\rhs}{\text{RHS}}
\newcommand{\pworminfty}{P_{\infty}'}
\newcommand{\piworminfty}{\pi'_{\infty}}
\newcommand{\pworminftybar}{\overline{P'}_{\infty}}
\newcommand{\piworminftybar}{\overline{\pi'}_{\infty}}
\newcommand{\fp}{F_{\mathbb{H}}}

\begin{frontmatter}
\title{A worm algorithm for the fully-packed loop model}
\author{Wei Zhang}
\address{Department of Physics,  Jinan University,\\ Guangzhou 510630, China}
\author{Timothy M. Garoni\corauthref{cor}}
\corauth[cor]{Corresponding author.}
\address{ARC Centre of Excellence for Mathematics and Statistics of Complex Systems, \\
  Department of Mathematics and Statistics, The University of Melbourne, \\Victoria~3010, Australia}
\ead{t.garoni@ms.unimelb.edu.au}
\author{Youjin Deng}
\address{Hefei National Laboratory for Physical Sciences at Microscale,\\ Department of Modern Physics, University of Science and Technology of China,
  \\ Hefei, 230027, China}
\ead{yjdeng@ustc.edu.cn}

\begin{abstract}
We present a Markov-chain Monte Carlo algorithm of {\em worm} type
that correctly simulates the fully-packed loop model with $n=1$ on the honeycomb lattice,
and we prove that it is ergodic and has uniform stationary distribution.
The honeycomb-lattice fully-packed loop model with $n=1$
is equivalent to the zero-temperature triangular-lattice antiferromagnetic Ising model,
which is fully frustrated and notoriously difficult to simulate.
We test this worm algorithm numerically and
estimate the dynamic exponent $z_{\rm exp} =0.515 (8)$.
We also measure several static quantities of interest, including loop-length and face-size moments.
It appears numerically that the face-size moments are governed by the magnetic dimension for percolation.
\end{abstract}

\date{14 November 2008}
\end{frontmatter}

\section{Introduction}
\label{intro}
The antiferromagnetic Ising model on the triangular lattice is of long-standing interest since it
provides a canonical example of geometric frustration: it is
topologically impossible to simultaneously minimize the interaction
energies of all three edges of an elementary triangular face.  
Recall that the Ising model on finite graph $G=(V,E)$ is defined by the measure
\begin{equation}
\mu_{G,\beta}(\sigma) \propto e^{-\beta H(\sigma)}, \qquad \sigma \in \{-1,+1\}^V,
\label{Ising measure}
\end{equation}
where the Hamiltonian, $H$, for the zero-field nearest-neighbor Ising model is simply
\begin{equation}
H(\sigma) = - \sum_{ij\in E } \sigma_i \sigma_j.
\label{Ising Hamiltonian}
\end{equation}
The coupling $\beta>0$ ($\beta<0$) corresponds to a ferromagnetic (antiferromagnetic) interaction.
We say an edge $ij$ is {\em satisfied} if the Ising interaction energy of its two
endpoints, $-\beta \sigma_i \sigma_j$, is minimized.  An edge is therefore
satisfied if its endpoints have parallel (anti-parallel) Ising spins
in the ferromagnetic (antiferromagnetic) case.  It is clear that a
given elementary face of the triangular lattice can have at most two
satisfied edges in an antiferromagnetic Ising model.  
The zero-field triangular-lattice antiferromagnetic Ising model has an exponentially large number of ground state degeneracies, leading to non-vanishing entropy per
spin~\cite{Wannier50}.
Although the model is disordered at all finite temperatures, at zero temperature
the two-point correlation function decays
algebraically~\cite{Stephenson64}, and so the model has a zero-temperature critical point.
By generalizing (\ref{Ising Hamiltonian}) to include anisotropic couplings,
dilution, longer range interactions, or a (staggered) magnetic field, rich phase diagrams have been
observed~\cite{NienhuisHilhorstBlote,BloteNightingale93,QuierozDomany94,QianBlote04,QianWegewijsBlote04,RastelliReginaTassi05,XiaYaoLiu07}.

Frustrated systems are notoriously difficult to simulate.
Naive algorithms such as single-spin-flip dynamics are inefficient at low temperatures, and become 
non-ergodic\!\!\!~\footnote{Following the typical usage in the physics literature, we take {\em ergodic} as synonymous with
{\em irreducible}.
Recall that a Markov chain is {\em irreducible} if for each pair of states $i$ and $j$ there is a positive probability that starting
in $i$ we eventually visit $j$, and vice versa.}
at zero temperature.
Indeed, even the best cluster algorithms~\cite{ZhangYang94,CoddingtonHan94,KandelBenAvDomany92} for simulating low temperature
frustrated Ising models become non-ergodic at zero temperature, although ergodicity can supposedly be obtained by augmenting the
cluster dynamics with single-spin-flip dynamics~\cite{ZhangYang94,CoddingtonHan94} 
(a similar hybrid approach is applied to the {\em string} dynamics discussed in~\cite{DharChaudhuriDasgupta00}).
Cluster algorithms have defined the dominant paradigm
for efficient Monte Carlo simulations of critical lattice models ever since the seminal work of Swendsen and Wang~\cite{SwendsenWang87}.
A more recent idea which is showing great promise
however is the idea of {\em worm} algorithms, first discussed in the context of classical spin models in~\cite{ProkofevSvistunov01}
(see also \cite{JerrumSinclair93}). The key idea behind the worm algorithm is to
simulate the high-temperature graphs of the spin model, considered
as a statistical-mechanical model in their own right.
The worm algorithm for the Ising model was recently studied in some
detail in \cite{DengGaroniSokal07c}, and it was observed to possess some unusual
dynamic features (see also~\cite{Wolff08}). Indeed, despite its local nature, the worm
algorithm was shown to be extraordinarily efficient \---
comparable to or better than the Swendsen-Wang (SW) method \--- for simulating some
aspects of the critical three-dimensional Ising
model\!\!\!~\footnote{In our opinion, the {\em conventional wisdom} that local algorithms are {\em a priori}
less efficient than cluster algorithms does not bear scrutiny.
Both the worm algorithm and the Sweeny algorithm~\cite{Sweeny83} (a local algorithm for simulating the
random-cluster model) have efficiencies~\cite{DengGaroniSokal07c,DengGaroniSokal07b}
comparable to, and in some instances better than, cluster algorithms
such as SW or the Chayes-Machta algorithm \cite{ChayesMachta98}.
Indeed, both algorithms can display {\em critical speeding-up}~\cite{DengGaroniSokal07b} in certain situations,
(although admittedly the Sweeny algorithm suffers from some algorithmic complications due to the need
for efficient cluster-finding subroutines).}. 
Given this success, a natural question to ask is whether
one can devise a valid worm algorithm to simulate a fully frustrated model such as 
the zero-temperature triangular-lattice antiferromagnetic Ising model.
The short answer is {\em yes}, as we demonstrate in Section~\ref{Worm algorithms},
although it does require a little thought.

Worm algorithms provide a natural way to simulate {\em Eulerian-subgraph models}.
Given a finite graph $G=(V,E)$, we call a bond configuration $A\subseteq E$ {\em Eulerian} if
every vertex in the subgraph $(V,A)$ has even degree (i.e., every vertex has an even number of incident bonds; zero is allowed).
The set of all such bond configurations defines the {\em cycle space} of $G$, denoted $\mathcal{C}(G)$.
Perhaps the simplest class of Eulerian-subgraph model is defined on the cycle space of a finite graph $G=(V,E)$, for $n,w>0$, by the probability measure
\begin{equation}
\phi_{G,w,n}(A) \propto \, n^{c(A)}\,w^{|A|},\qquad A\in\mathcal{C}(G),
\label{eulerian-subgraph measure}
\end{equation}
where $c(A)$ is the cyclomatic number
of the spanning subgraph $(V,A)$.  Note that on graphs of maximum
degree $\le3$ the only possible Eulerian subgraphs consist of
a collection of disjoint cycles, or {\em loops},  and $c(A)$ is then simply the number of such loops.  Consequently,
Eulerian-subgraph models often go by the name of {\em loop models},
and the honeycomb lattice, being a $3$-regular graph, has played a
distinguished role in the literature on such loop models
\cite{Nienhuis82,Nienhuis84}. These geometric models play a major
role in recent developments of conformal field theory
\cite{DiFrancescoMathieuSenechal97} via their connection with Schramm
Loewner evolution (SLE) \cite{Schramm00,RohdeSchramm05,Lawler05}.  

In this work we focus on the case $n=1$, and we write
$\phi_{G,w}:=\phi_{G,w,1}$.  In this case it can be seen that
(\ref{eulerian-subgraph measure}) corresponds to an Ising model on
$G$, and also to an Ising model on the dual graph $G^*$, when $G$ is
planar.  Indeed, it is an elementary exercise to derive the following two identities
relating the partition functions of the Ising and Eulerian-subgraph models
\begin{align}
Z^{\text{Ising}}_{G,\beta} &= (2^{|V|}\cosh^{|E|}\beta)\, Z^{\text{Eulerian}}_{G,\tanh(\beta)}
\label{high T Z}
\\
Z^{\text{Ising}}_{G^*,\beta} &= (2e^{\beta|E|})\, Z^{\text{Eulerian}}_{G,e^{-2\beta}}.
\label{low T Z}
\end{align}
Note that (\ref{high T Z}) corresponds to an Eulerian-subgraph model with positive weights only when $\beta>0$, i.e. in the ferromagnetic case, and that only
the region $0\le w \le 1$ is covered by the correspondence. 
By contrast, (\ref{low T Z}) gives positive weights for all $\beta\in\mathbb{R}$ and corresponds to the whole region $0\le~w\le+\infty$.
Since $\tanh(\beta)$ is small when $\beta$ is small the relation (\ref{high T Z}) is commonly referred to as the {\em high-temperature} expansion of the Ising model.
Similarly, since $e^{-2\beta}$ is small when $\beta$ is large, i.e. in the ferromagnetic regime at low temperatures, 
(\ref{low T Z}) is commonly referred to as the {\em low temperature} expansion of the Ising model. 
However since $e^{-2\beta}$ is in fact large in the antiferromagnetic regime at low temperatures we shall refrain from using this terminology. 
The $n=1$ Eulerian-subgraph model with $0\le w \le 1$ thus
corresponds to ferromagnetic Ising models on both $G$ and $G^*$, while the
model with $1\le~w\le+\infty$ corresponds to an antiferromagnetic
Ising model on $G^*$. 
In particular, we point out that the honeycomb-lattice loop model with $w>1$ corresponds to the triangular-lattice antiferromagnetic Ising model.

The honeycomb-lattice loop model exhibits an interesting phase diagram that has been studied in detail~\cite{BloteNienhuis89,BloteNienhuis94}.
When $n=1$ it undergoes an Ising phase transition at $w_c =1/\sqrt{3}$ from a disordered phase when $w<w_c$ to 
a {\em densely-packed} phase when $w_c < w < +\infty$.
Interestingly, the entire region  $w_c < w < +\infty$ displays critical
behavior, and the model is in the two-dimensional percolation universality class.
The case $w=+\infty$ is of especial interest, and is the focus of this article.
In this case $\phi_{G,w}$ is simply uniform measure
on the set of {\em fully-packed} configurations, i.e. the set of all Eulerian subgraphs with the maximum possible number of edges,
and the model is referred to as the fully-packed loop (FPL) model. It should be emphasized that the FPL model is critical, although it is in a different 
universality class to the densely-packed phase~\cite{BloteNienhuis94}.
Every fully-packed bond configuration is such that each vertex is visited by precisely one loop, i.e. each vertex has degree $2$.
Therefore only two thirds of the edges are occupied in any given fully-packed configuration;
this fact can be seen as a symptom of the frustration of the triangular-lattice antiferromagnetic Ising model.
Indeed, according to (\ref{low T Z}) the FPL model corresponds to the triangular-lattice antiferromagnetic Ising model at {\em zero} temperature.
Fig.~\ref{fully-packed honeycomb torus} shows a typical fully-packed configuration.
\begin{figure}
  \caption{\label{fully-packed honeycomb torus}Typical fully-packed configuration on the honeycomb lattice with periodic boundary conditions. 
    Thick lines denote occupied edges, thin lines denote vacant edges.}
  \begin{center}
    \psset{unit=0.5cm}
    \begin{pspicture}(0,0)(10,8)
      \pcline[linewidth=0.1](0,1)(4,1)\pcline[linewidth=0.02](4,1)(5,1)\pcline[linewidth=0.1](5,1)(10,1)
      \pcline[linewidth=0.1](0,2)(3,2)\pcline[linewidth=0.02](3,2)(4,2)\pcline[linewidth=0.1](4,2)(10,2)
      \pcline[linewidth=0.1](0,3)(3,3)\pcline[linewidth=0.02](3,3)(4,3)\pcline[linewidth=0.1](4,3)(10,3)
      \pcline[linewidth=0.1](0,4)(4,4)\pcline[linewidth=0.02](4,4)(5,4)\pcline[linewidth=0.1](5,4)(10,4)
      \pcline[linewidth=0.1](0,5)(5,5)\pcline[linewidth=0.02](5,5)(6,5)\pcline[linewidth=0.1](6,5)(10,5)
      \pcline[linewidth=0.1](0,6)(5,6)\pcline[linewidth=0.02](5,6)(6,6)\pcline[linewidth=0.1](6,6)(10,6)
      \pcline[linewidth=0.02](0,1)(0,2)\pcline[linewidth=0.02](0,3)(0,4)\pcline[linewidth=0.02](0,5)(0,6)
      \pcline[linewidth=0.02](2,1)(2,2)\pcline[linewidth=0.02](2,3)(2,4)\pcline[linewidth=0.02](2,5)(2,6)
      \pcline[linewidth=0.1](4,1)(4,2)\pcline[linewidth=0.1](4,3)(4,4)\pcline[linewidth=0.02](4,5)(4,6)
      \pcline[linewidth=0.02](6,1)(6,2)\pcline[linewidth=0.02](6,3)(6,4)\pcline[linewidth=0.1](6,5)(6,6)
      \pcline[linewidth=0.02](8,1)(8,2)\pcline[linewidth=0.02](8,3)(8,4)\pcline[linewidth=0.02](8,5)(8,6)
      \pcline[linewidth=0.02](1,0)(1,1)\pcline[linewidth=0.02](1,2)(1,3)\pcline[linewidth=0.02](1,4)(1,5)\pcline[linewidth=0.02](1,6)(1,7)
      \pcline[linewidth=0.02](3,0)(3,1)\pcline[linewidth=0.1](3,2)(3,3)\pcline[linewidth=0.02](3,4)(3,5)\pcline[linewidth=0.02](3,6)(3,7)
      \pcline[linewidth=0.1](5,0)(5,1)\pcline[linewidth=0.02](5,2)(5,3)\pcline[linewidth=0.1](5,4)(5,5)\pcline[linewidth=0.1](5,6)(5,7)
      \pcline[linewidth=0.02](7,0)(7,1)\pcline[linewidth=0.02](7,2)(7,3)\pcline[linewidth=0.02](7,4)(7,5)\pcline[linewidth=0.02](7,6)(7,7)
      \pcline[linewidth=0.02](9,0)(9,1)\pcline[linewidth=0.02](9,2)(9,3)\pcline[linewidth=0.02](9,4)(9,5)\pcline[linewidth=0.02](9,6)(9,7)
    \end{pspicture}
  \end{center}
\end{figure}
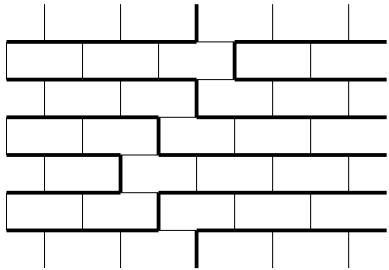

The essence of the worm idea is to enlarge a configuration space of Eulerian bond configurations to
include a pair of {\em defects} (i.e., vertices of odd degree), and then to move these defects via
random walk. When the two defects coincide, the configuration becomes Eulerian once more.
In the standard worm algorithm we view the simulation as a simulation
of high-temperature graphs of the Ising model on $G$ defined by (\ref{high T Z}).
This interpretation is only valid for $0\le w \le 1$.
However, another useful interpretation when $G$ is planar is that the worm algorithm simulates an Ising model on
$G^*$, and this interpretation is valid for all $w>0$. We shall return
to this point in some detail in Section~\ref{Relation to Eulerian-subgraph and Ising models}. 
We wish to emphasize here however that if we have a worm algorithm to simulate the FPL model, we immediately have an algorithm to simulate the zero-temperature
triangular-lattice antiferromagnetic Ising model.

Unfortunately, devising a valid worm algorithm for simulating the FPL model
is not simply a case of taking the large $w$ limit of the ``standard'' version of the worm algorithm, as presented in~\cite{DengGaroniSokal07c}.
Indeed, as the bond weight $w$ increases, the efficiency of the worm algorithm
presented in Refs.~\cite{ProkofevSvistunov01,DengGaroniSokal07c} drops
rapidly, because the random walker moves ever more slowly. In the limit $w \rightarrow \infty$ the random walk becomes completely
frozen, and  the standard worm algorithm becomes invalid (the details
will be  explained in Section \ref{Worm algorithms}).  In this work,
we present a variation of the worm algorithm presented
in~\cite{DengGaroniSokal07c} which efficiently simulates the honeycomb-lattice FPL model, when $n=1$.
Importantly, we prove rigorously that this algorithm is
ergodic, and has uniform stationary distribution, on the fully-packed configurations.
We have tested this worm algorithm numerically, and we estimate the dynamic exponent $z_{\exp}=0.515(8)$. (See Section~\ref{worm dynamic data} for
a precise definition of $z_{\exp}$.)

The organization of the current work is as follows.  Section~\ref{Worm algorithms}
reviews the standard worm algorithm~\cite{ProkofevSvistunov01,DengGaroniSokal07c} for Ising high-temperature graphs, and then introduces a version to simulate
the FPL model. In Section~\ref{worm results} we present the results of our simulations of the FPL model using the
worm algorithm discussed in Section~\ref{Worm algorithms}. Finally, Section~\ref{discussion} contains a discussion.

\section{Worm algorithms}
\label{Worm algorithms}
We begin with a review of the standard worm algorithm defined on an arbitrary graph, which essentially follows the presentation in
\cite{DengGaroniSokal07c}, and then go on to discuss its relationship to the Eulerian-subgraph model on $G$ as well as Ising models on $G$ and $G^*$.
After demonstrating why the standard version becomes non-ergodic as $w\to\infty$, we then present a valid worm algorithm for simulating
the honeycomb-lattice FPL model.

\subsection{The ``standard'' worm algorithm}
Fix a finite graph $G=(V,E)$, and for any $A\subseteq E$ let $\partial A\subseteq V$ denote the set of all vertices which have odd degree in
the spanning subgraph $(V,A)$.  Loosely, 
$\partial A$ is just the set of sites that {\em touch} an
odd number of the bonds in the bond configuration $A$.  If $u,v\in V$ are distinct we write
\begin{equation*}
\mathcal{S}_{u,v}:=\{A\subseteq E\,:\, \partial A=\{u,v\}\},
\end{equation*}
and
\begin{equation*}
\mathcal{S}_{v,v}:=\{A\subseteq E\,:\,
\partial A=\emptyset\}.
\end{equation*}
We emphasize that $\mathcal{S}_{v,v}=\mathcal{C}(G)$ for every $v\in V$.
We take the state space of the worm algorithm to be
\begin{equation*}
\mathcal{S}:=\{(A,u,v) \,:\, u,v\in V \text{ and } A\in
\mathcal{S}_{u,v}\},
\end{equation*}
i.e., all ordered triples $(A,u,v)$ with $A \subseteq E$ and $u,v \in V$, 
such that $A~\in~\mathcal{S}_{u,v}$. Note that if
$(A,u,v)\in\mathcal{S}$ then $A$ is Eulerian iff $u=v$. Thus the bond
configurations allowed in the state space of the worm algorithm constitute a
superset of the Eulerian configurations.
Finally, we assign probabilities to the configurations in $\mathcal{S}$ according to
\begin{equation}
\piworm(A,u,v)\propto  d_u\,d_v\, w^{|A|},\qquad (A,u,v)\in \mathcal{S},
\label{worm measure}
\end{equation}
where $d_v$ denotes the degree in $G$ of $v\in V$. In the following, when we wish to refer to the degree of $v\in V$ in the 
spanning subgraph $(V,A)$ we will write $d_v(A)$. Loosely, $d_v(A)$ is simply the number of bonds that touch $v$ in the bond configuration $A$. 
In this notation we have $d_v=d_v(E)$.

The first step in constructing the standard worm algorithm is to consider
the worm {\em proposal matrix}, $P^{(0)}$, which is defined for all $uu'\in E$ and $v\in V$ by
\begin{equation}
P^{(0)}[(A,u,v)\to(A\triangle uu',u',v)] = P^{(0)}[(A,v,u)\to(A\triangle uu',v,u')] = \frac{1}{2d_u},
\label{worm proposal}
\end{equation}
all other entries being zero. Here $\triangle$ denotes symmetric difference, i.e. delete the bond $uu'$ from $A$ if it is present, or insert it if it is absent.
It is easy to see that $P^{(0)}$ is an irreducible transition matrix on $\mathcal{S}$.
According to~(\ref{worm proposal}) the moves proposed by the worm algorithm are as follows:
Pick uniformly at random one of the two defects (say, $v$) and one of
the edges emanating from $v$ (say, $vv'$), then move from the current configuration $(A,u,v)$ to the new configuration $(A~\triangle~vv',u,v')$.

Now we simply use the usual Metropolis-Hastings prescription (see e.g.~\cite[\S 4]{SokalLectures}) to
assign acceptance probabilities to the moves proposed by $P^{(0)}$, so that the resulting transition matrix, $\pworm$, is in detailed balance with (\ref{worm measure}).
Explicitly, for all $uu'\in E$ and $v\in V$ we have
\begin{equation}
\begin{split}
   \pworm[(A,u,v)\to(A\triangle uu',u',v)] 
&= \pworm[(A,v,u)\to(A\triangle uu',v,u')] \\
&= 
\frac{1}{2d_u}
\begin{cases}
F(w)   & uu'\not\in A\\
F(1/w) & uu'    \in A\\
\end{cases}
\end{split}
\label{worm P}
\end{equation}
where $F:[0,+\infty]\to[0,1]$ is any function satisfying
\begin{equation}
F(z) = z\,F(1/z) \text{ for all } z.
\label{F equation}
\end{equation}
Two concrete examples of such $F$ which are commonly used in practice are $F(z)=\min(1,z)$ and $F(z)=z/(1+z)$.
For a given choice of $F$, the transitions (\ref{worm P}) define $\pworm$ uniquely since
all other transitions occur with zero probability except the identity transitions $(A,u,v)\to(A,u,v)$, whose transition probabilities are fixed by normalization to be
\begin{equation}
\begin{split}
\pworm[(A,u,v)\to(A,u,v)]
= 
1 
&- F(w)   \left[ 1 - \left(\frac{d_u(A)}{2d_u}+\frac{d_v(A)}{2d_v}\right)\right]\\
&- F(1/w) \left(\frac{d_u(A)}{2d_u}+\frac{d_v(A)}{2d_v}\right).\\
\end{split}
\label{identity transitions}
\end{equation}
For any choice of $F$, one can easily verify that $\pworm$ and $\piworm$ are in detailed balance.

\subsection{Relation to Eulerian-subgraph and Ising models}
\label{Relation to Eulerian-subgraph and Ising models}
A natural question to ask at this stage is what precisely is the
relationship between $\phi_{G,w}$ and the worm transition matrix
(\ref{worm P})?  To address this question, let us consider the Markov
chain induced on the subset
\begin{equation}
S:=\{(A,v,v)\in\mathcal{S}\}\subset\mathcal{S},
\label{eulerian subspace}
\end{equation}
in which the bond configurations are Eulerian. More precisely, let's
suppose that we only observe the worm chain when it is in a state in $S$.
This defines a new Markov chain, a single step of
which corresponds in the old chain to the transition (not necessarily
in one step) from a state $(A,v,v)$ to another state $(A',v',v')$. The
new transition probability to move from  $(A,v,v)$ to $(A',v',v')$ is
found by computing the probability that the original chain starting in
$(A,v,v)$ hits $S$ for the first time at state $(A',v',v')$. This is
the probability that the chain goes from $(A,v,v)$ to $(A',v',v')$ in
one step (which is zero unless $A=A'$ and $v=v'$),  plus the
probability that it goes to a state outside $S$ and then re-enters $S$
for the first time at  $(A',v',v')$. 
A nice discussion of this general problem can be found in
\cite[\S 6.1]{KemenySnell76}, including a proof of
\begin{lemma}
\label{sub-state space chain}
Let $P$ be an irreducible transition matrix on a finite state space $\mathcal{S}$
with stationary distribution $\pi$. Define a new Markov chain by
only observing the original chain corresponding to $P$ when it visits a state in $S\subset\mathcal{S}$. The new chain
is an irreducible Markov chain on $S$ with
stationary distribution
$$ 
\overline{\pi}_s = \frac{\pi_s}{\sum_{s'\in S}\pi_{s'}}, \qquad s \in S.
$$
\end{lemma}

As a consequence of Lemma~\ref{sub-state space chain} the worm Markov chain restricted to the Eulerian subspace
(\ref{eulerian subspace}) has a stationary distribution
$\piwormbar$ given explicitly by
\begin{equation}
\piwormbar(A,v,v)=\left(\frac{d_v^2}{\sum_{v'\in V}\,d_{v'}^2}\right)\,\phi_{G,w}(A).
\label{restricted worm measure}
\end{equation}
Consequently we have
\begin{equation}
\langle X\rangle_{\piwormbar} = \langle X \rangle_{\phi_{G,w}}
\label{worm and Eulerian expectations are equal}
\end{equation}
for any observable $X:\mathcal{C}(G)\to\mathbb{R}$ of the original
Eulerian-subgraph model, and hence we can indeed use the worm
algorithm to simulate $\phi_{G,w}$.

We note that when $G$ is planar (\ref{worm and Eulerian expectations are equal}) also implies that the worm algorithm correctly simulates 
the Ising model on $G^*$ considered in (\ref{low T Z}).
Indeed, suppose that  $G$ is planar with dual $G^*=(V^*,E^*)$, and consider the two-to-one correspondence $\sigma\mapsto A_{\sigma}$ 
from $\{-1,+1\}^{V^*}\to\mathcal{C}(G)$ where
\begin{equation}
A_{\sigma}:=\{ij\in E \,:\, \sigma_{i^*}\neq\sigma_{j^*}\}.
\label{dual mapping}
\end{equation}
In words, for any spin configuration on $G^*$ we draw on $G$ the boundaries of the spin domains.
It is an elementary exercise to show that for all $\sigma\in \{-1,+1\}^{V^*}$ we have
\begin{equation}
\phi_{G,w}(A_{\sigma}) = 2\mu_{G^*,\beta}(\sigma), \qquad w = e^{-2\beta},
\label{low T measure relation}
\end{equation}
where $\mu_{G^*,\beta}$ is the mass function of the Ising model on $G^*$, as defined in (\ref{Ising measure}).
We emphasize that although (\ref{low T measure relation}) is often called a {\em low temperature} representation, 
it is an exact result valid for all $-\infty\le\beta \le +\infty$, or equivalently for all $0\le w\le +\infty$. From (\ref{low T measure relation}) we see explicitly 
that for any Ising observable $Y: \{-1,+1\}^{V^*}\to\mathbb{R}$ that is even under global spin flips 
(which is the case for all observables of physical interest in zero field)
we have $\langle Y \rangle_{\mu_{G,\beta}} = \langle X \rangle_{\phi_{G,w}}$ where $w=e^{-2\beta}$ and $X:\mathcal{C}(G)\to\mathbb{R}$ is defined by
$X(A_{\sigma}) = Y(\sigma) = Y(-\sigma)$. Consequently (\ref{worm and Eulerian expectations are equal}) does indeed allow us to simulate the Ising model on $G^*$
using the worm algorithm.

We should also mention that when $w\le 1$ the worm algorithm can be used to simulate properties related to the two-point correlation function of the 
Ising model on $G$ defined by (\ref{high T Z}). Indeed, it is straightforward to generalize (\ref{high T Z}) to obtain an expansion for the 
two-point correlation function
\begin{equation}
Z^{\text{Ising}}_{G,\beta}\langle \sigma_u \sigma_v\rangle_{\mu_{G,\beta}} = \sum_{A\in \mathcal{S}_{u,v}} (\tanh\beta)^{|A|}.
\label{Two-point function relation}
\end{equation}
As an example of the use of (\ref{Two-point function relation}), consider the observable $\mathcal{D}_0$ on $\mathcal{S}$ defined so that
\begin{equation}
\label{D_0 definition}
\mathcal{D}_0(A,u,v) = \delta_{u,v}.
\end{equation}
In other words $\mathcal{D}_0$ is the indicator for being in $S$.
It is straightforward to show that provided $G$ is regular we have
$$\langle\mathcal{D}_0 \rangle_{\piworm}=V/\langle \mathcal{M}^2\rangle_{\mu_{G,\beta}},$$
where $\mathcal{M}=\sum_{v\in V}\sigma_v$ is the magnetization and the $\piworm$ expectations use $w=\tanh\beta$.
In particular, in a translationally invariant system
\begin{equation}
\label{chi D_0 relation}
\langle \mathcal{D}_0 \rangle_{\piworm} = 1/\chi^{\text{Ising}}_{G,\beta}.
\end{equation}

Thus when $w\le1$ the worm algorithm simulates both an Ising model on $G$ and an Ising model on $G^*$. Quantities like $\mathcal{D}_0$ depend on the full Markov chain
on $\mathcal{S}$, and so if one's interest is to obtain quantities related to the two-point function for the Ising model on $G$ with $w=\tanh(\beta)$ then one 
must consider the full Markov chain. However, if one's interest is to compute properties of the Eulerian-subgraph model (\ref{eulerian-subgraph measure}), or the 
corresponding Ising model on $G^*$ with $\beta=e^{-2\beta}$, then one is only interested in the Markov chain induced on $S$. It is the latter models that are
our interest in the present work, and we emphasize that in this case the restriction $w\le 1$ does not apply.

We note, finally, that \cite{Wang05} uses ideas similar to Lemma~\ref{sub-state space chain} and (\ref{low T measure relation})
to simulate a low temperature Ising spin glass with a worm algorithm.

\subsection{Periodic boundary conditions}
\label{periodic boundary conditions}
For completeness, we now briefly address the question of the effect of boundary conditions when $G$ is a regular lattice. To illustrate, 
we consider $G=\mathbb{H}$, where 
$\mathbb{H}$ denotes a finite subgraph of the honeycomb lattice drawn on a torus as
in Fig.~\ref{fully-packed honeycomb torus}. The periodic boundary conditions imply that $\mathbb{H}$ is non-planar, however we can still construct
the dual lattice $\mathbb{T}$ in the usual way, and it is easy to see that $\mathbb{T}$ is simply a finite piece of the triangular lattice also drawn on a torus.
It is now no longer the case however that every $A\in\mathcal{C}(\mathbb{H})$ defines the domain boundaries of an Ising spin configuration; indeed 
Fig.~\ref{fully-packed honeycomb torus} provides 
an example for which no consistent assignment of Ising spins is possible. Suppose however that we let $\mathcal{C}^+(\mathbb{H})$ denote the set of 
all $A\in\mathcal{C}(\mathbb{H})$
which wind the torus an {\em even} number of times in both directions. For these configurations there is no ambiguity in assigning Ising 
configurations according to the correspondence (\ref{dual mapping}), and it is easy to see that (\ref{dual mapping}) defines a 
two-to-one correspondence from $\{-1,+1\}^{V(\mathbb{T})}$ onto $\mathcal{C}^+(\mathbb{H})$.
It is easy to generalize (\ref{low T measure relation}) to show that it is now replaced by
\begin{equation}
\label{periodic Ising measure}
\mu_{\mathbb{T},\beta}(\sigma)
=
\frac{1}{2}
\frac{(e^{-2\beta})^{|A_{\sigma}|}}{\sum_{A'\in \mathcal{C}^+(\mathbb{H})}(e^{-2\beta})^{|A'|}}.
\end{equation}
In addition, if one applies Lemma~\ref{sub-state space chain} to the subspace $\mathcal{C}^+(\mathbb{H})$ then we obtain
\begin{equation}
\label{even winding worm measure}
\overline{\pi}_{\mathbb{H},w}(A,v,v)
=
\frac{1}{V}
\frac{w^{|A|}}{\sum_{A'\in \mathcal{C}^+(\mathbb{H})}w^{|A'|}}, \qquad \text{ for all } A\in \mathcal{C}^+(\mathbb{H}).
\end{equation}
Combining (\ref{periodic Ising measure}) and (\ref{even winding worm measure}) we see immediately that 
$\overline{\pi}_{\mathbb{H},e^{-2\beta}}(A_{\sigma},v,v)=(2/V)\mu_{\mathbb{T},\beta}(\sigma)$.
Therefore if we simulate a worm chain on $\mathbb{H}$ with coupling $e^{-2\beta}$ 
and only measure this chain when it is both Eulerian {\em and} winds the torus an even number of times, then we are effectively simulating the 
Ising model on $\mathbb{T}$ at inverse temperature $\beta$.

\subsection{A worm algorithm for the honeycomb lattice FPL model}
Thus far we have glossed over an important issue, namely the
irreducibility of the worm transition matrix $\pworm$.  It is not hard to see
that $\pworm$ is irreducible whenever $F(w)$ and $F(1/w)$ are both strictly positive.
Problems arise as $w\to\infty$ however, since it is easy to show that if $F:[0,+\infty]\to[0,1]$ satisfies (\ref{F equation}) then $F(0)=0$.
Consequently, as $w\to\infty$ the probabilities for transitions that remove an edge vanish.
Indeed, all states $(A,u,v)\in\mathcal{S}$ for which both $d_u(A)=d_u$ and $d_v(A)=d_v$ become absorbing as 
$w\to\infty$. This is easy to see from (\ref{identity transitions}) since such states have $\pworm[(A,u,v)\to(A,u,v)]=1-F(1/w)$.

Suppose now that $G$ is $k$-regular, i.e. all vertices have degree $k$. Then $(A,u,v)$ will be absorbing when $w=+\infty$ iff $d_u(A)=d_v(A)=k$.
Recall that if $(A,u,v)\in\mathcal{S}$ then when $u=v$ the vertex degree $d_u(A)=d_v(A)$ is even, whereas when $u\neq v$ both $d_u(A)$ and $d_v(A)$ are odd.
Thus if $k$ is even then $(A,u,v)$ can be absorbing only if $u=v$ whereas if $k$ is odd $(A,u,v)$ can be absorbing only if $u\neq v$.
Therefore when $k$ is odd all states $(A,v,v)$ with Eulerian $A$ remain non-absorbing;
$(A,v,v)\rightarrow(A,v,v)$ occurs with probability $d_v(A)/k<1$ when $w=+\infty$.
In particular, on the honeycomb lattice we can now see that as $w\to \infty$ all states $(A,v,v)$ with Eulerian $A$ remain non-absorbing while all states
$(A,u,v)$ with $u\neq v$ and $d_u(A)=d_v(A)=3$ become absorbing.
Therefore once both defects have degree $3$ the chain remains in that state for eternity.

How do we resolve this problem?  A simple answer
is to avoid this trap of endless identity transitions by explicitly forbidding
$(A,u,v)\to(A,u,v)$ whenever $u\neq v$.  Since, when simulating Eulerian-subgraph models,
we only observe the chain when it visits an Eulerian state $(A,v,v)$ we
may hope that by only modifying the transitions from non-Eulerian states we may
recover irreducibility without sacrificing the correctness of the
stationary distribution. We shall see that this is indeed possible.

To this end we now define a new transition matrix, $\pworminfty$, which defines a valid Monte Carlo algorithm to
simulate the FPL model on the honeycomb lattice, i.e. when $G=\mathbb{H}$ with
$\mathbb{H}$ as defined in Section~\ref{periodic boundary conditions}.
To define the transition probabilities $\pworminfty[(A,v,v)\rightarrow\cdot\,]$
we simply take the limits of~(\ref{worm P})
\begin{equation}
\label{eulerian limit}
\pworminfty[(A,v,v)\to(A\cup vv',v',v)] = \pworminfty[(A,v,v)\to(A\cup vv',v,v')] = \frac{1}{6},
\end{equation}
and (\ref{identity transitions})
\begin{equation}
\label{identity limit}
\pworminfty[(A,v,v)\to(A,v,v)] = \frac{d_v(A)}{3}.
\end{equation}
All other transitions from $(A,v,v)$ are assigned zero probability; in particular, one cannot {\em remove} an edge from an Eulerian state.

To define the transition probabilities
$\pworminfty[(A,u,v)\rightarrow\cdot\,]$ with $u\neq v$ we use the
following simple rules: first, choose uniformly at random one of the
two defects, say $u$. Since $u\neq v$ we must have
$d_u(A)\in\{1,3\}$. If $d_u(A)=3$ we choose uniformly at random one of
the three occupied edges incident to $u$, say $uu'$,
and we delete it by making the transition $(A,u,v)\rightarrow(A\setminus uu',u',v)$. This ensures that we can
never get stuck when the defects are full \--- i.e. it removes the
problem of absorbing states suffered by the $w\to\infty$ limit of
$\pworm$.  If, on the other hand, $d_u(A)=1$ we choose uniformly
at random one of the two vacant edges incident to $u$, say $uu'$, and
occupy it by making the transition 
$(A,u,v)\rightarrow(A\cup uu',u',v)$. This guarantees that we cannot produce an isolated vertex
by moving a degree $1$ defect, which is obviously a desirable property
when one wants to simulate a fully-packed model. These rules correspond to the following transition probabilities when
$u\neq v$
\begin{align}
\pworminfty[(A,u,v)\to(A\triangle uu',u',v)] &= \pworminfty[(A,v,u)\to(A\triangle uu',v,u')]\nonumber\\
&=
\begin{cases}
1/6 & d_u(A) = 3,\\
1/4 & uu'\not\in A.
\end{cases}
\end{align}
All other transitions from $(A,u,v)$ with $u\neq v$ are assigned zero probability; in particular, no identity transitions are allowed.

While we hope that the above discussion convinces the reader that $\pworminfty$ provides a plausible (and natural) candidate for simulating 
the FPL model on the honeycomb lattice, we of course do not claim that it {\em proves} such an assertion. A proof of the validity of $\pworminfty$
is presented in Section~\ref{proofs}.

In terms of a Monte Carlo algorithm, $\pworminfty$ corresponds to Algorithm~\ref{fpl algorithm}. The abbreviation UAR simply means 
{\em uniformly at random}.
\begin{algorithm}
  \begin{algthm}[Honeycomb-lattice fully-packed loop model] $\,$
    \label{fpl algorithm}
    \begin{algorithmic}
      \LOOP
      \STATE Current state is $(A,u,v)$
      \IF{$u=v$}
      \STATE Choose, UAR, one of the 3 neighbors of $u$ (say $u'$)
      \IF{$uu'\not\in A$}
      \STATE Perform, UAR, either $(A,u,u)\to(A\cup uu',u',u)$ or $(A,u,u)\to(A\cup uu',u,u')$
      \ELSIF{$uu'\in A$}
      \STATE $(A,u,u)\to(A,u,u)$
      \ENDIF
      \ELSIF{$u\neq v$}
      \STATE Choose, UAR, one of the 2 defects (say $u$)
      \IF{$d_u(A)=3$}
      \STATE Choose, UAR, one of the 3 neighbors of $u$ (say $u'$)
      \STATE $(A,u,v)\to(A\setminus uu',u',v)$
      \ELSIF{$d_u(A)=1$}
      \STATE Choose, UAR, one of the 2 vacant edges incident to $u$ (say $uu'$)
      \STATE $(A,u,v)\to(A\cup uu',u',v)$
      \ENDIF
      \ENDIF
      \ENDLOOP
    \end{algorithmic}
  \end{algthm}
\end{algorithm}

\subsection{Proof of validity of Algorithm~\ref{fpl algorithm}}
\label{proofs}
This Section provides a rigorous proof of the validity of Algorithm~\ref{fpl algorithm}.
Readers uninterested in such details may simply choose to trust us and skip to the next Section.

Proving {\em validity} of Algorithm~\ref{fpl algorithm} boils down to showing that $\pworminfty$ is irreducible (in a suitable sense) and that it has the {\em right}
stationary distribution (in a suitable sense). With regard to the latter question we note that $\phi_{\mathbb{H},\infty}(A)=I_{\fp}(A)/|\fp|$, where
\begin{equation*}
\fp:=\{A\in\mathcal{C}(\mathbb{H})\,:\, d_v(A)=2 \text{ for all } v\in V(\mathbb{H})\}
\end{equation*}
and $I_{\fp}$ is its indicator. That is, $\phi_{\mathbb{H},\infty}$ is just uniform measure on the set $\fp$ of fully-packed configurations on $\mathbb{H}$.

Let us pause to recall some basic background regarding finite Markov chains (see e.g.~\cite{Iosifescu80,GrimmettStirzaker06}).
Consider then a Markov chain on a finite state space with transition matrix $P$.
We say state $i$ {\em communicates} with state $j$, and write $i\rightarrow j$, if the chain may ever visit state $j$
with positive probability, having started in state $i$. 
We say states $i$ and $j$ {\em intercommunicate}, and write $i\leftrightarrow j$, if $i\rightarrow j$ and $j\rightarrow i$.
A set of states $\mathcal{C}$ is called {\em irreducible} if $i\leftrightarrow j$ for all $i,j\in\mathcal{C}$, and it is called {\em closed} if $P_{ij}=0$
for all $i\in \mathcal{C}$ and $j\not\in\mathcal{C}$.
A state $i$ is {\em recurrent} if, with probability 1, the chain eventually returns to $i$, having started in $i$; and it is {\em transient} otherwise.
If every state in $\mathcal{C}$ is recurrent (transient) we say $\mathcal{C}$ itself is recurrent (transient). 
It can be shown that $\mathcal{C}$ is recurrent iff it is closed.

Now let us define the subset of states
\begin{equation}
\mathcal{R} = \{(A,u,v)\in\mathcal{S} \, : \, d_x(A) \neq 0 \text{ for all } x\}
\end{equation}
The set $\mathcal{R}$ thus consists of all those states with {\em no isolated vertices},
and is where all the action takes place when considering $\pworminfty$. We emphasize that 
the set of all bond configurations $A$ for which $(A,v,v)\in\mathcal{R}$ corresponds precisely with 
$\fp$. 
\begin{proposition}
\label{closed proposition}
$\mathcal{R}$ is closed.
\end{proposition}
\begin{proposition}
\label{irreducibility-transience proposition}
$\mathcal{R}$ is irreducible and $\mathcal{S}\setminus\mathcal{R}$ is transient.
\end{proposition}
Thus when running Algorithm~\ref{fpl algorithm} we are free to
begin in any state in $\mathcal{S}$, and (due to the transience of $\mathcal{S}\setminus\mathcal{R}$)
with probability 1 the chain will end up inside $\mathcal{R}$, from where
(due to $\mathcal{R}$ being closed) the chain then never leaves.
Furthermore, (due to the irreversibility of $\mathcal{R}$) all states in $\mathcal{R}$ will eventually be visited.
Finally, we have the following explicit form for the stationary distribution of $\pworminfty$.
\begin{proposition}
\label{stationarity proposition}
The unique stationary distribution of $\pworminfty$ is $\piworminfty$ where
\begin{equation*}
\piworminfty
(A,u,v)
=
\begin{cases}
0 & (A,u,v)\not\in\mathcal{R}\\
\lambda & (A,u,v)\in\mathcal{R}, u=v\\
2\lambda/3 & (A,u,v)\in\mathcal{R}, u\sim v, d_u(A)=d_v(A)=3\\
\lambda/9 & (A,u,v)\in\mathcal{R}, u\sim v, uv \not\in A, d_u(A)=d_v(A)=1\\
2\lambda/9 & (A,u,v)\in\mathcal{R}, u\sim v, uv \in A, d_u(A)=d_v(A)=1\\
2\lambda/9 & (A,u,v)\in\mathcal{R}, u\neq v, u\not\sim v, d_u(A)+d_v(A)=2\\
\lambda/3 & (A,u,v)\in\mathcal{R}, u\neq v, u\not\sim v, d_u(A)+d_v(A)=4\\
\lambda/2 & (A,u,v)\in\mathcal{R}, u\neq v, u\not\sim v, d_u(A)+d_v(A)=6\\
\end{cases}
\end{equation*}
and $\lambda$ is finite and positive.
\end{proposition}
In particular, $\piworminfty$ is constant on the states $(A,v,v)\in\mathcal{R}$.
It follows that if we consider the Markov chain constructed by only measuring the $\pworminfty$ chain when the
defects coincide, then Lemma~\ref{sub-state space chain} implies that $\pworminftybar$ has stationary distribution
\begin{align}
\piworminftybar(A,v,v) 
&= \frac{\piworminfty(A,v,v)}{\sum_{(A',v',v')\in\mathcal{S}}\piworminfty(A',v',v')}\\
&= \frac{\phi_{F_{\mathbb{H}},\infty}(A)}{V},
\end{align}
as desired. Consequently $\langle X \rangle_{\piworminftybar} = \langle X \rangle_{\phi_{\mathbb{H},\infty}}$ 
for any observable $X:F_{\mathbb{H}}\to\mathbb{R}$ of the FPL model.

We now conclude this Section with proofs of Propositions~\ref{closed proposition}, \ref{irreducibility-transience proposition}, and \ref{stationarity proposition}.
\begin{proof}[Proof of Proposition~\ref{closed proposition}]
All the states in $\mathcal{S}\setminus\mathcal{R}$ have at least one isolated vertex, while the states in $\mathcal{R}$ have none.
Since $\pworminfty$ only allows transitions that add/remove at most one edge, the only states $(A,u,v)\in\mathcal{R}$ which could possibly make a transition to
a state with an isolated vertex are those with at least one vertex $u$ with $d_u(A)=1$, and the transition would need to remove the edge $uu'\in A$.
However, we have
$$
\pworminfty[(A,u,v)\rightarrow(A\setminus uu',u',v)]=0.
$$
In fact, if $d_u(A)=1$ the only non-zero $\pworminfty[(A,u,v)\rightarrow\cdot\,]$ that correspond to the removal of an edge are of the form
$$
\pworminfty[(A,u,v)\rightarrow(A\setminus vv',u,v')]=1/6>0
$$
with $d_v(A)=3$. Therefore the only possible way a transition could remove $uu'$ was if $v=u'$ and we made the transition
$(A,u,v)\rightarrow(A\setminus uv,u,u)$. See Fig.~\ref{closed proof figure}.
\begin{figure}
  \caption{\label{closed proof figure}Here thick lines denote occupied
    edges, thin lines denote vacant edges, while dashed lines denote
    edges whose occupation status is undecided. Periodic boundary
    conditions are imposed. A transition capable of creating an isolated vertex can only occur from a state for which the neighborhoods of the defects are as shown.
    There clearly exist bond configurations in $\mathcal{S}$ with
    defect neighborhoods as shown,
    however Lemma~\ref{neighboring defects lemma} implies that no
    such bond configurations exist in $\mathcal{R}$.
  }
  \begin{center}
    \psset{unit=0.5cm}
    \begin{pspicture}(0,0)(10,7)
      {\psset{linestyle=dashed,dash=3pt 2pt,linewidth=0.01}
	\pcline(0,1)(10,1)\pcline(0,2)(10,2)\pcline(0,3)(10,3)
	\pcline(0,4)(2,4)\pcline(5,4)(10,4)
	\pcline(0,5)(10,5)\pcline(0,6)(10,6)
	\pcline(0,1)(0,2)\pcline(0,3)(0,4)\pcline(0,5)(0,6)
	\pcline(2,1)(2,2)\pcline(2,3)(2,4)\pcline(2,5)(2,6)
	\pcline(4,1)(4,2)                 \pcline(4,5)(4,6)
	\pcline(6,1)(6,2)\pcline(6,3)(6,4)\pcline(6,5)(6,6)
	\pcline(8,1)(8,2)\pcline(8,3)(8,4)\pcline(8,5)(8,6)
	\pcline(1,0)(1,1)\pcline(1,2)(1,3)\pcline(1,4)(1,5)\pcline(1,6)(1,7)
	\pcline(3,0)(3,1)\pcline(3,2)(3,3)                 \pcline(3,6)(3,7)
	\pcline(5,0)(5,1)\pcline(5,2)(5,3)\pcline(5,4)(5,5)\pcline(5,6)(5,7)
	\pcline(7,0)(7,1)\pcline(7,2)(7,3)\pcline(7,4)(7,5)\pcline(7,6)(7,7)
	\pcline(9,0)(9,1)\pcline(9,2)(9,3)\pcline(9,4)(9,5)\pcline(9,6)(9,7)
      }
      \pcline[linewidth=0.15](3,4)(5,4)\pcline[linewidth=0.15](4,3)(4,4)\pcline[linewidth=0.03](2,4)(3,4)\pcline[linewidth=0.03](3,4)(3,5)
      \qdisk(3,4){0.2}\qdisk(4,4){0.2}
      \uput[270](3,4){$u$}\uput[90](4,4){$v$}
    \end{pspicture}
  \end{center}
\end{figure}
Such a transition would indeed occur with positive probability. However Lemma~\ref{neighboring defects lemma}
guarantees that there do not exist any states $(A,u,v)\in\mathcal{R}$ with $u\sim v$ and $d_u(A)\neq d_v(A)$. 
Therefore
$$
\pworminfty[(A,u,v)\rightarrow(A',u',v')]=0
$$
whenever $(A,u,v)\in\mathcal{R}$ and $(A,u,v)\in\mathcal{S}\setminus\mathcal{R}$.
\end{proof}

\begin{proof}[Proof of Proposition~\ref{irreducibility-transience proposition}]
Let $H\in\fp$ denote the fully-packed configuration in which every horizontal edge is occupied and every vertical edge is vacant.
We begin by proving that every state in $\mathcal{S}$ communicates with $(H,w,w)\in\mathcal{R}$ for some $w$.
We make frequent use of the lemmas listed in Appendix~\ref{irreducibility lemmas appendix}.

Suppose then that $(A,u,v)\in\mathcal{S}$. We can generate a new state from $(A,u,v)$ via the map $f:\mathcal{S}\to\mathcal{S}$
with $f(A,u,v)$ defined by the following prescription:
\begin{algorithmic}
  \IF{$d_u(A)=0,1$}
  \STATE Choose a vacant horizontal edge $uu'$
  \RETURN $(A\cup uu',u',v)$
  \ELSIF{$d_u(A)=3$}
  \STATE Choose the occupied vertical edge $uu'$
  \RETURN $(A\setminus uu',u',v)$
  \ELSIF{$d_u(A)=2$}
  \IF{there are no vacant horizontal edges}
  \RETURN $(H,u,u)$
  \STATE // Note that it must be the case that $(A,u,v)=(H,u,u)$
  \ELSE
  \STATE Choose a vacant horizontal edge $ww'$ for which $(A,u,u)\rightarrow(A,w,w)$
  \STATE // Lemma~\ref{transient lemma} guarantees that such a $ww'$ exists
  \RETURN $(A\cup ww',w',w)$
  \ENDIF
  \ENDIF
\end{algorithmic}
The first observation to make is that for any $(A,u,v)\in\mathcal{S}$ we have $(A,u,v)\rightarrow f(A,u,v)$.
Indeed, if $d_u(A)=0,1$ or 3 we simply have
$$\pworminfty[(A,u,v)\rightarrow f(A,u,v)]>0.$$
If $d_u(A)=2$ and there are no vacant horizontal edges then it must be the case that $(A,u,v)=(H,u,u)=f(H,u,u)$, so
$(A,u,v)\leftrightarrow f(A,u,v)$ follows trivially.
Finally, if $d_u(A)=2$ and there exists at least one vacant horizontal edge
then Lemma~\ref{transient lemma} guarantees that at least one such edge $ww'$ satisfies $(A,u,u)\rightarrow(A,w,w)$ and since
$$
\pworminfty[(A,w,w)\rightarrow(A\cup ww',w',w)]=1/6
$$
it follows that $(A,u,u)\rightarrow(A\cup ww',w',w)$.
So we indeed have $(A,u,v)\rightarrow f(A,u,v)$ for any $(A,u,v)\in\mathcal{S}$, and in fact
$(A,u,v)\rightarrow f^n(A,u,v)$ for any $n\in\mathbb{N}$, where $f^n$ denotes $n$-fold composition of $f$ with itself, i.e.
$$f^n=f\circ f\circ \dots \circ f$$
$n$ times.

Now, whenever $(A,u,v)\neq (H,w,w)$ for some $w$, the state $f(A,u,v)$ has either one less occupied vertical edge, or one more occupied
horizontal edge, than $(A,u,v)$. Therefore, since there are only a finite number of horizontal and vertical edges, if we start in any $(A,u,v)\in\mathcal{S}$ 
and apply $f$ repeatedly then we must eventually have
$f^n(A,u,v)=(H,w,w)$ for some $w$, with $n$ necessarily {\em finite}. 
See for example Fig.~\ref{irreducibility proof figure}.
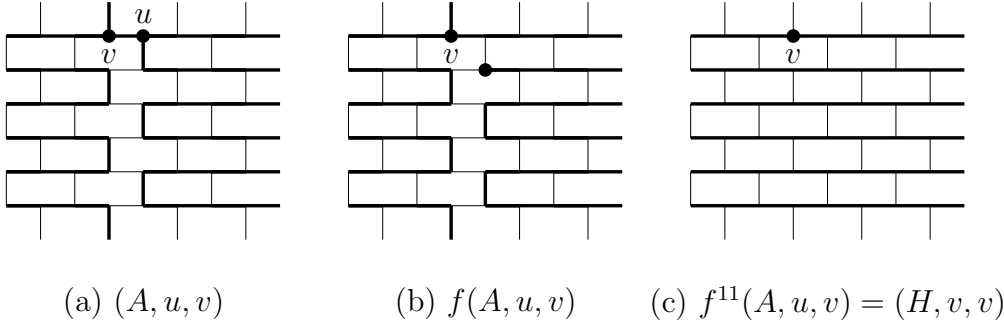
\begin{figure}
  \caption{\label{irreducibility proof figure}Example of repeated
    application of $f$ to a configuration $(A,u,v)$. Starting from the
    initial configuration in ({a}), application of $f$ removes
    an occupied vertical edge resulting in the configuration $f(A,u,v)$ shown in
    ({b}). Continuing in this way, alternately adding vacant horizontal
    edges and removing occupied vertical edges, finally results in the
    configuration $f^{11}(A,u,v)=(H,v,v)$ shown in ({c}).
    Thick lines denote occupied edges, thin lines denote vacant edges.
    Periodic boundary conditions are imposed.
  }
  \begin{center}
    \psset{unit=0.45cm}
    \begin{pspicture}(0,-2)(30,8)
      %% (A,v,v)
      \pcline[linewidth=0.02](0,1)(8,1)\pcline[linewidth=0.02](0,2)(8,2)\pcline[linewidth=0.02](0,3)(8,3)
      \pcline[linewidth=0.02](0,4)(8,4)\pcline[linewidth=0.02](0,5)(8,5)\pcline[linewidth=0.02](0,6)(8,6)
      \pcline[linewidth=0.1](0,1)(3,1)\pcline[linewidth=0.1](4,1)(8,1)
      \pcline[linewidth=0.1](0,2)(3,2)\pcline[linewidth=0.1](4,2)(8,2)
      \pcline[linewidth=0.1](0,3)(3,3)\pcline[linewidth=0.1](4,3)(8,3)
      \pcline[linewidth=0.1](0,4)(3,4)\pcline[linewidth=0.1](4,4)(8,4)
      \pcline[linewidth=0.1](0,5)(3,5)\pcline[linewidth=0.1](4,5)(8,5)
      \pcline[linewidth=0.1](0,6)(8,6)
      \pcline[linewidth=0.1](0,1)(1,1)\pcline[linewidth=0.1](0,2)(1,2)\pcline[linewidth=0.1](0,3)(1,3)\pcline[linewidth=0.1](0,4)(1,4)\pcline[linewidth=0.1](0,5)(1,5)\pcline[linewidth=0.1](0,6)(1,6)
      \pcline[linewidth=0.1](2,1)(3,1)\pcline[linewidth=0.1](2,2)(3,2)\pcline[linewidth=0.1](2,3)(3,3)\pcline[linewidth=0.1](2,4)(3,4)\pcline[linewidth=0.1](2,5)(3,5)\pcline[linewidth=0.1](2,6)(3,6)
      \pcline[linewidth=0.1](4,1)(5,1)\pcline[linewidth=0.1](4,2)(5,2)\pcline[linewidth=0.1](4,3)(5,3)\pcline[linewidth=0.1](4,4)(5,4)\pcline[linewidth=0.1](4,5)(5,5)\pcline[linewidth=0.1](4,6)(5,6)
      \pcline[linewidth=0.1](6,1)(7,1)\pcline[linewidth=0.1](6,2)(7,2)\pcline[linewidth=0.1](6,3)(7,3)\pcline[linewidth=0.1](6,4)(7,4)\pcline[linewidth=0.1](6,5)(7,5)\pcline[linewidth=0.1](6,6)(7,6)
      \pcline[linewidth=0.02](0,1)(0,2)\pcline[linewidth=0.02](0,3)(0,4)\pcline[linewidth=0.02](0,5)(0,6)
      \pcline[linewidth=0.02](2,1)(2,2)\pcline[linewidth=0.02](2,3)(2,4)\pcline[linewidth=0.02](2,5)(2,6)
      \pcline[linewidth=0.1](4,1)(4,2)\pcline[linewidth=0.1](4,3)(4,4)\pcline[linewidth=0.1](4,5)(4,6)
      \pcline[linewidth=0.02](6,1)(6,2)\pcline[linewidth=0.02](6,3)(6,4)\pcline[linewidth=0.02](6,5)(6,6)
      \pcline[linewidth=0.02](1,0)(1,1)\pcline[linewidth=0.02](1,2)(1,3)\pcline[linewidth=0.02](1,4)(1,5)\pcline[linewidth=0.02](1,6)(1,7)
      \pcline[linewidth=0.1](3,0)(3,1)\pcline[linewidth=0.1](3,2)(3,3)\pcline[linewidth=0.1](3,4)(3,5)\pcline[linewidth=0.1](3,6)(3,7)
      \pcline[linewidth=0.02](5,0)(5,1)\pcline[linewidth=0.02](5,2)(5,3)\pcline[linewidth=0.02](5,4)(5,5)\pcline[linewidth=0.02](5,6)(5,7)
      \pcline[linewidth=0.02](7,0)(7,1)\pcline[linewidth=0.02](7,2)(7,3)\pcline[linewidth=0.02](7,4)(7,5)\pcline[linewidth=0.02](7,6)(7,7)
      \qdisk(3,6){0.2}\uput[270](3,6){$v$}
      \qdisk(4,6){0.2}\uput[90](4,6){$u$}
      \uput[270](4,-1){(a) $(A,u,v)$}

      %% f(A,u,v)
      \pcline[linewidth=0.02](10,1)(18,1)\pcline[linewidth=0.02](10,2)(18,2)\pcline[linewidth=0.02](10,3)(18,3)
      \pcline[linewidth=0.02](10,4)(18,4)\pcline[linewidth=0.02](10,5)(18,5)\pcline[linewidth=0.02](10,6)(18,6)
      \pcline[linewidth=0.1](10,1)(13,1)\pcline[linewidth=0.1](14,1)(18,1)
      \pcline[linewidth=0.1](10,2)(13,2)\pcline[linewidth=0.1](14,2)(18,2)
      \pcline[linewidth=0.1](10,3)(13,3)\pcline[linewidth=0.1](14,3)(18,3)
      \pcline[linewidth=0.1](10,4)(13,4)\pcline[linewidth=0.1](14,4)(18,4)
      \pcline[linewidth=0.1](10,5)(13,5)\pcline[linewidth=0.1](14,5)(18,5)
      \pcline[linewidth=0.1](10,6)(18,6)
      \pcline[linewidth=0.1](10,1)(11,1)\pcline[linewidth=0.1](10,2)(11,2)\pcline[linewidth=0.1](10,3)(11,3)\pcline[linewidth=0.1](10,4)(11,4)\pcline[linewidth=0.1](10,5)(11,5)\pcline[linewidth=0.1](10,6)(11,6)
      \pcline[linewidth=0.1](12,1)(13,1)\pcline[linewidth=0.1](12,2)(13,2)\pcline[linewidth=0.1](12,3)(13,3)\pcline[linewidth=0.1](12,4)(13,4)\pcline[linewidth=0.1](12,5)(13,5)\pcline[linewidth=0.1](12,6)(13,6)
      \pcline[linewidth=0.1](14,1)(15,1)\pcline[linewidth=0.1](14,2)(15,2)\pcline[linewidth=0.1](14,3)(15,3)\pcline[linewidth=0.1](14,4)(15,4)\pcline[linewidth=0.1](14,5)(15,5)\pcline[linewidth=0.1](14,6)(15,6)
      \pcline[linewidth=0.1](16,1)(17,1)\pcline[linewidth=0.1](16,2)(17,2)\pcline[linewidth=0.1](16,3)(17,3)\pcline[linewidth=0.1](16,4)(17,4)\pcline[linewidth=0.1](16,5)(17,5)\pcline[linewidth=0.1](16,6)(17,6)
      \pcline[linewidth=0.02](10,1)(10,2)\pcline[linewidth=0.02](10,3)(10,4)\pcline[linewidth=0.02](10,5)(10,6)
      \pcline[linewidth=0.02](12,1)(12,2)\pcline[linewidth=0.02](12,3)(12,4)\pcline[linewidth=0.02](12,5)(12,6)
      \pcline[linewidth=0.1](14,1)(14,2)\pcline[linewidth=0.1](14,3)(14,4)\pcline[linewidth=0.02](14,5)(14,6)
      \pcline[linewidth=0.02](16,1)(16,2)\pcline[linewidth=0.02](16,3)(16,4)\pcline[linewidth=0.02](16,5)(16,6)
      \pcline[linewidth=0.02](11,0)(11,1)\pcline[linewidth=0.02](11,2)(11,3)\pcline[linewidth=0.02](11,4)(11,5)\pcline[linewidth=0.02](11,6)(11,7)
      \pcline[linewidth=0.1](13,0)(13,1)\pcline[linewidth=0.1](13,2)(13,3)\pcline[linewidth=0.1](13,4)(13,5)\pcline[linewidth=0.1](13,6)(13,7)
      \pcline[linewidth=0.02](15,0)(15,1)\pcline[linewidth=0.02](15,2)(15,3)\pcline[linewidth=0.02](15,4)(15,5)\pcline[linewidth=0.02](15,6)(15,7)
      \pcline[linewidth=0.02](17,0)(17,1)\pcline[linewidth=0.02](17,2)(17,3)\pcline[linewidth=0.02](17,4)(17,5)\pcline[linewidth=0.02](17,6)(17,7)
      \qdisk(13,6){0.2}\uput[270](13,6){$v$}
      \qdisk(14,5){0.2}
      \uput[270](14,-1){(b) $f(A,u,v)$}

      %% f^{11}(A,u,v)
      \pcline[linewidth=0.1](20,1)(28,1)
      \pcline[linewidth=0.1](20,2)(28,2)
      \pcline[linewidth=0.1](20,3)(28,3)
      \pcline[linewidth=0.1](20,4)(28,4)
      \pcline[linewidth=0.1](20,5)(28,5)
      \pcline[linewidth=0.1](20,6)(28,6)

      \pcline[linewidth=0.02](20,1)(20,2)\pcline[linewidth=0.02](20,3)(20,4)\pcline[linewidth=0.02](20,5)(20,6)
      \pcline[linewidth=0.02](22,1)(22,2)\pcline[linewidth=0.02](22,3)(22,4)\pcline[linewidth=0.02](22,5)(22,6)
      \pcline[linewidth=0.02](24,1)(24,2)\pcline[linewidth=0.02](24,3)(24,4)\pcline[linewidth=0.02](24,5)(24,6)
      \pcline[linewidth=0.02](26,1)(26,2)\pcline[linewidth=0.02](26,3)(26,4)\pcline[linewidth=0.02](26,5)(26,6)
      \pcline[linewidth=0.02](21,0)(21,1)\pcline[linewidth=0.02](21,2)(21,3)\pcline[linewidth=0.02](21,4)(21,5)\pcline[linewidth=0.02](21,6)(21,7)
      \pcline[linewidth=0.02](23,0)(23,1)\pcline[linewidth=0.02](23,2)(23,3)\pcline[linewidth=0.02](23,4)(23,5)\pcline[linewidth=0.02](23,6)(23,7)
      \pcline[linewidth=0.02](25,0)(25,1)\pcline[linewidth=0.02](25,2)(25,3)\pcline[linewidth=0.02](25,4)(25,5)\pcline[linewidth=0.02](25,6)(25,7)
      \pcline[linewidth=0.02](27,0)(27,1)\pcline[linewidth=0.02](27,2)(27,3)\pcline[linewidth=0.02](27,4)(27,5)\pcline[linewidth=0.02](27,6)(27,7)
      \qdisk(23,6){0.2}\uput[270](23,6){$v$}

      \uput[270](24,-1){(c) $f^{11}(A,u,v)=(H,v,v)$}
    \end{pspicture}
  \end{center}
\end{figure}
It then immediately follows that $(A,u,v)\rightarrow(H,w,w)$.

Suppose now that $(A,u,v)\in\mathcal{S}\setminus\mathcal{R}$. 
As we have just shown, there is at least one state $(H,w,w)\in\mathcal{R}$ with which $(A,u,v)$ communicates,
i.e. $(A,u,v)\to(H,w,w)$, and there is thus a non-zero probability that
starting in $(A,u,v)$ a finite number of transitions will take us to $(H,w,w)$. 
But since $(H,w,w)\in\mathcal{R}$ and Proposition~\ref{closed proposition} tells us that $\mathcal{R}$ is closed, there is zero probability of ever 
leaving $\mathcal{R}$ again, and in particular there is zero probability of ever returning to $(A,u,v)\in\mathcal{S}\setminus\mathcal{R}$.
There is therefore a non-zero probability that starting in $(A,u,v)\in\mathcal{S}\setminus\mathcal{R}$
we never return to $(A,u,v)$. Therefore the state $(A,u,v)$ is transient and it follows at once that in fact the whole space $\mathcal{S}\setminus\mathcal{R}$
is transient.

Now let us turn our attention to the irreducibility of $\mathcal{R}$.
It is clear that $f(A,u,v)\in\mathcal{R}$ whenever $(A,u,v)\in\mathcal{R}$. Furthermore, whenever $(A,u,v)\in\mathcal{R}$ we have
$f(A,u,v)\leftrightarrow(A,u,v)$.
To see this we note: we can never have $d_u(A)=0$ when $(A,u,v)\in\mathcal{R}$;
if $d_u(A)=1$ then Lemma~\ref{degree 1 intercommunication} implies $(A,u,v)\leftrightarrow f(A,u,v)$; 
if $d_u(A)=3$ then Lemma~\ref{degree 3 intercommunication} implies $(A,u,v)\leftrightarrow f(A,u,v)$; 
if $d_u(A)=2$ then Lemma~\ref{eulerian intercommunication} implies $(A,u,u)\leftrightarrow(A,w,w)$ for all $w$, and if $ww'$ is vacant
Lemma~\ref{quasi-eulerian irreducibility} implies that $(A,w,w)\leftrightarrow(A\cup ww',w',w)$, so that $(A,u,u)\leftrightarrow(A\cup ww',w',w)$.

Therefore we now see that for any $(A,u,v)\in\mathcal{R}$ we have
$(A,u,v)\leftrightarrow f(A,u,v)$, and indeed $(A,u,v)\leftrightarrow f^n(A,u,v)$ for any $n\in\mathbb{N}$.
Since, as argued above, we must have $(H,w,w)=f^n(A,u,v)$ for some $w$ and finite $n$, it immediately follows that 
$(A,u,v)\leftrightarrow(H,w,w)$. Since every $(A,u,v)\in\mathcal{R}$ intercommunicates with $(H,w,w)\in\mathcal{R}$ for some (in fact all) $w$ it follows that
$\mathcal{R}$ is irreducible.
\end{proof}

\begin{remark}
The careful reader will notice that there is some ambiguity in the
definition of  $f$ presented in the proof of
Proposition~\ref{irreducibility-transience proposition}. For instance, if there
is more than one vacant  horizontal edge which one should we choose?
Such careful readers can easily construct an appropriate rule to make
the choice of this edge precise (or make the choice of edge random and
view $f$ as a random variable). The validity of the proof is
independent of any such technical details and so we have deliberately
swept such issues under the proverbial rug.
\end{remark}

\begin{proof}[Proof of Proposition~\ref{stationarity proposition}]
It is straightforward (if a little tedious) to prove that $\piworminfty$ is a stationary distribution for $\pworminfty$
by simply considering each of the eight cases in the definition of $\piworminfty$,
explicitly computing the right-hand side of
\begin{equation*}
\piworminfty(A,u,v) 
= \sum_{(B,x,y)\in\mathcal{S}}\piworminfty(B,x,y)\pworminfty[(B,x,y)\to(A,u,v)]
\end{equation*}
and verifying that it equals the left-hand side, for every $(A,u,v)\in\mathcal{S}$. We omit the  details.

Clearly, the constant $\lambda$ appearing in the definition of $\piworminfty$ must be chosen so that $\sum_{(A,u,v)\in\mathcal{S}}\piworminfty(A,u,v)=1$,
but its exact value is not really of any concern to us.
We simply observe that it is some well defined finite positive number. Indeed it is elementary to derive the upper and lower bounds
$1/\lambda\ge V|\fp|>0$ and $1/\lambda \le |\mathcal{R}|$.

Since $\mathcal{S}$ has only one closed irreducible set of states, $\mathcal{R}$, it can have only one stationary distribution, so $\piworminfty$ is unique.
\end{proof}

\section{Numerical results}
\label{worm results}
We simulated the FPL model on an $L \times L $ honeycomb lattice with periodic boundary conditions using
Algorithm~\ref{fpl algorithm}. We studied fourteen different system sizes in the range $6 \leq L \leq 900$, each being a multiple of $3$.

\subsection{Observables measured}
We measured the following observables in our simulations.
All observables were measured only when the defects coincided, except for $\mathcal{D}_0$ which was measured every step.

\begin{itemize}
\item The number of loops $\mathcal{N}_l$ (cyclomatic number)
\item The mean-square loop length
\begin{equation}
\mathcal{L}_2 := \sum_{i=1}^{\mathcal{N}_l} (\text{length of } i^{th}\text{ loop})^2
\label{def_quantity_l2}
\end{equation}
\item The sum of the $n$th powers of the face sizes
\begin{equation}
\mathcal{\mathcal{G}}_n := \sum_{f} |f|^n
\end{equation}
Every $A\in\fp$ can be decomposed into a number of {\em faces}, each consisting of a collection of elementary hexagons, such that
every pair of neighboring elementary hexagons in $\mathbb{H}$ which share an unoccupied edge in $A$ belong to the same face.
The size $|f|$ of face $f$ is then simply the number of elementary hexagons which it contains.
We considered $n=2$ and $n=4$.
\item $\mathcal{D}_0$ as defined in~(\ref{D_0 definition})
\end{itemize}
 
From these observables we computed the following quantities:
\begin{itemize}
\item The loop-number density $n_l:=\langle \mathcal{N}_l\rangle/L^2$ 
\item The loop-number fluctuation $C_l := \text{var}(\mathcal{N}_l)/L^2$
\item The (normalized) expectation of $\mathcal{L}_2$
$$L_2:=\langle\mathcal{L}_2\rangle/L^2$$
\item The (normalized) expectation of $\mathcal{G}_2$ and $\mathcal{G}_4$
\begin{equation}
\begin{split}
G_2&:=\frac{1}{L^2}\langle\mathcal{G}_2\rangle\\
G_4&:=\frac{1}{L^4}\langle\mathcal{G}_4\rangle
\end{split}
\label{def_quantity_g2}
\end{equation}
\item The ratio $Q_g := G_2^2/G_4$
\item The mean number of iterations of the full worm chain between visits to the Eulerian subspace
\begin{equation*}
T_E:=1/\langle\mathcal{D}_0\rangle_{\piworminfty}
\end{equation*}

\end{itemize}

\begin{remark}
In the case of the FPL model the number of bonds $\mathcal{N}(A)=|A|$ is constant since every vertex has degree 2,
and so $\mathcal{N}$ is a trivial observable in this case, unlike the case for Ising high-temperature graphs~\cite{DengGaroniSokal07c}.
\end{remark}

\subsection{Static data} 
For each observable $\mathcal{O}=T_E,L_2,G_2$ we performed a least-squares fit to the simple finite-size scaling ansatz
\begin{equation*}
\mathcal{O}(L) = L^{d-2 X_\mathcal{O}} (\mathcal{O}_0 +\mathcal{O}_1 L^{2X_\mathcal{O}-d} +\mathcal{O}_2 L^{y_1} +\mathcal{O}_3 L^{y_2}).
\end{equation*}
The $\mathcal{O}_1$ term arises from the regular part of the free energy, while the $\mathcal{O}_2$ and $\mathcal{O}_3$ terms correspond
to corrections to scaling. The correction-to-scaling exponents were fixed to $y_1=-2$ and $y_2=-3$, and of course $d=2$.
As a precaution against corrections to scaling we impose a lower cutoff
$L\ge L_{\text{min}}$ on the data points admitted to the fit, and we studied
systematically the effects on the fit of varying the value of
$L_{\text{min}}$. We estimate
\begin{equation}
\begin{split}
\label{D0 L2 G2 fits}
X_{T_E}&=0.2499(2),\\
X_{L_2}&=0.2498(4),\\
X_{G_2}&=0.1040(3).
\end{split}
\end{equation}
According to~\cite{BloteNienhuis94} the magnetic scaling dimension of the $n=1$ FPL model is
$X_h = 1/4$. From (\ref{D0 L2 G2 fits}) we therefore conjecture that in fact
\begin{equation}
\label{L2 and D0 exponent}
X_{T_E}=X_{L_2}=X_h=1/4.
\end{equation}
In particular, we expect the number of iterations of Algorithm~\ref{fpl algorithm} between visits to the Eulerian subspace to scale like
$L^{2-2X_h}=L^{3/2}$.

We remark that (\ref{D0 L2 G2 fits}) suggests $X_{G_2}$ is very close (perhaps equal) to $X_{h}^{\rm perc} = 5/48$,
the magnetic scaling dimension for models in the two-dimensional percolation universality class.
Here is a hand-waving argument suggesting that in fact $X_{G_2}=X_{h}^{\rm perc}$ might be an identity:
For general $n$ it is known~\cite{BloteNienhuis94} that
the honeycomb-lattice loop model defined by (\ref{eulerian-subgraph measure}) displays simultaneously the universal properties of a densely-packed loop model
with loop fugacity $n$ and those of a model with central charge $c=1$ and thermal dimension $X_t=1$. 
The zero-temperature triangular-lattice antiferromagnetic Ising model has $c=1$ and $X_t=1$, and when $n=1$ the densely-packed loop model 
is in the percolation universality class. We may therefore expect the $n=1$ FPL model to display some of the critical behavior of percolation.

In Fig.~\ref{fig_l2} we plot the data for $L_2/L^2$ and $T_E/L^2$ versus $L^{-2X_{h}}$, and in Fig.~\ref{fig_g2} we plot
$G_2/L^2$ versus $L^{-2X_{h}^{\rm perc}}$.

\begin{figure}
\begin{center}
\includegraphics[scale=0.8]{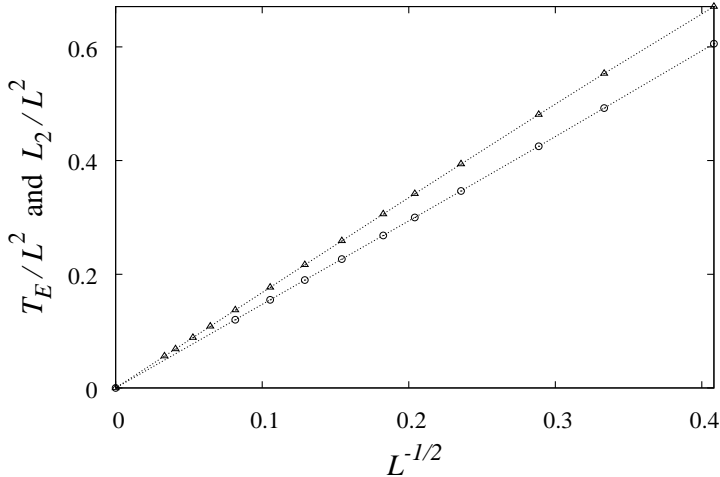}
\caption{\label{fig_l2} Plot of $L_2/L^2$ and $T_E/L^2$, represented by $\bigcirc$ and
$\triangle$ respectively, versus $L^{-2X_{h}}=L^{-1/2}$.
Error bars are smaller than the size of the symbols. The dashed lines are simply to guide the eye.}
\end{center}
\end{figure}
\begin{figure}
\begin{center}
\includegraphics[scale=0.8]{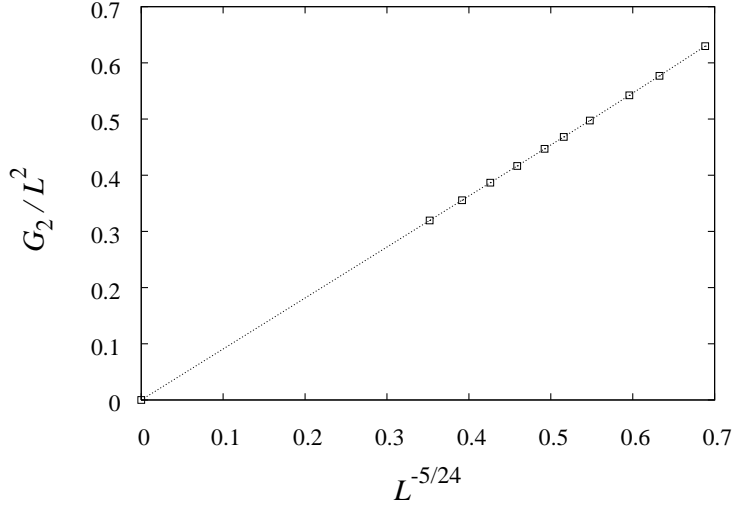}
\caption{\label{fig_g2} Plot of $G_2/L^2$ versus $L^{-2X_h^{\rm perc}}=L^{-5/24}$. 
Error bars are smaller than the size of the symbols. The dashed lines are simply to guide the eye.}
\end{center}
\end{figure}

The data for $n_l$ and $C_l$ were fitted to the ansatz
\begin{equation}
\mathcal{O}(L) = \mathcal{O}_0 +\mathcal{O}_1 L^{y_1} +\mathcal{O}_2 L^{y_2},
\label{constant ansatz}
\end{equation}
with the exponents $y_1$ and $y_2$ fixed to $-2$ and $-4$ respectively.
We estimated $\mathcal{O}_0=0.028836 (2)$ for $n_l$ and and  $\mathcal{O}_0=0.02620(3)$ for $C_l$.  
In Fig.~\ref{fig_nc} we plot $n_l$ and $C_l$ versus $L^{-2}$.
\begin{figure}
\begin{center}
\includegraphics[scale=0.8]{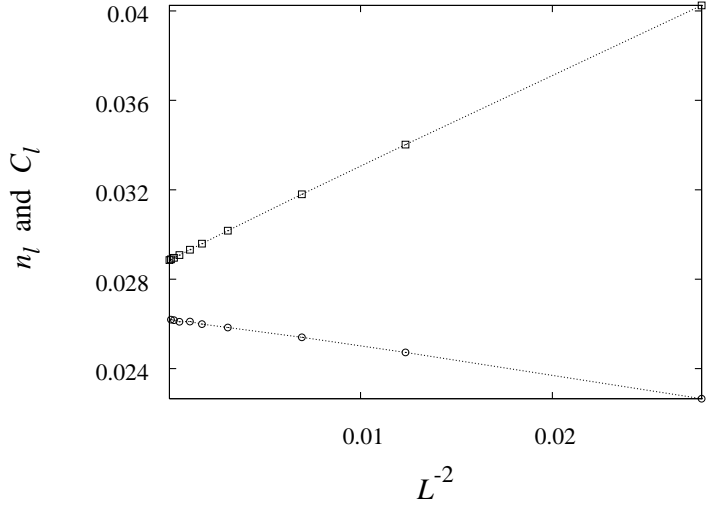}
\caption{\label{fig_nc} Plot of $n_l$ and $C_l$, represented by $\Box$ and
$\bigcirc$ respectively, versus $L^{-2}$. Error bars are smaller than
the size of the symbols. The dashed lines are simply to guide the eye.}
\end{center}
\end{figure}

Finally, we fit the data for the dimensionless ratio $Q_g$ to
(\ref{constant ansatz}) with fixed exponents $y_1=-2$ and $y_2=-4$,
and with an additional correction term proportional to 
$L^{2X_{h}^{\rm perc}-2}$. We estimate $\mathcal{O}_0=1.0248(4)$.

\subsection{Dynamic data}
\label{worm dynamic data}
For any observable $\mathcal{O}$, we define its autocorrelation function 
\begin{equation*}
\rho_{\mathcal{O}}(t) := \langle \mathcal{O}(t)\mathcal{O}(0)\rangle - \langle \mathcal{O}\rangle^2,
\end{equation*}
where $\langle \cdot \rangle$ denotes expectation with respect to the stationary distribution.
We then define the corresponding exponential autocorrelation time
\begin{equation}
   \tau_{{\rm exp},\mathcal{O}} :=\limsup_{t \to \pm\infty} 
\frac{|t|}{- \log |\rho_{\mathcal{O}}(t)|},
\end{equation}
and integrated autocorrelation time
\begin{equation}
   \tau_{{\rm int},\mathcal{O}}:= \frac{1}{2}\, \sum_{t = -\infty}^{\infty} \rho_{\mathcal{O}}(t) \;.
\end{equation}
Typically, all observables $\mathcal{O}$ (except those that, for symmetry
reasons, are ``orthogonal'' to the slowest mode) have the same exponential autocorrelation time,
so $\tau_{{\rm exp},\mathcal{O}} = \tau_{{\rm exp}}$. However, they may have
very different amplitudes of ``overlap'' with this slowest mode; in
particular, they may have very different values of the integrated
autocorrelation time, which controls the efficiency of Monte Carlo
simulations \cite{SokalLectures}.  

The autocorrelation times typically diverge as a critical point is
approached, most  often like $\tau \sim \xi^z$, where $\xi$ is the
spatial correlation length and $z$ is a dynamic exponent. 
This phenomenon is referred to as {\rm critical slowing-down}~\cite{HohenbergHalperin77,SokalLectures}.
More precisely, we define dynamic critical exponents $z_{\rm exp}$ and $z_{{\rm int},\mathcal{O}}$ by 
\begin{equation}
\begin{split}
\label{tau definitions}
\tau_{\rm exp} 
&\sim \xi^{z_{\rm exp}},\\
\tau_{{\rm int},\mathcal{O}} 
&\sim \xi^{z_{{\rm int},\mathcal{O}}}.
\end{split}
\end{equation}
On a finite lattice at criticality, $\xi$ can here be replaced by $L$.

During the simulations we measured the observables (except for $\mathcal{D}_0$)
only when the chain visited the Eulerian subspace, roughly every $T_E\sim L^{d-2X_h}$ iterations, or
{\em hits}, of Algorithm~\ref{fpl algorithm}. However,
it is natural when defining $z_{\exp}$ and $z_{\text{int},\mathcal{O}}$ via (\ref{tau definitions}) to measure time in units of
{\em sweeps} of the lattice, i.e. $L^d$ hits. Since one sweep takes of order $L^{2X_h}$ visits to the Eulerian subspace, in units of
``visits to the Eulerian subspace'' we have $\tau\sim L^{z+2X_h}$.

For each observable $\mathcal{O}=\mathcal{N}_l$, $\mathcal{D}_0$, $\mathcal{L}_2$, $\mathcal{G}_2$
we computed $\rho_{\mathcal{O}}(t)$ and $\tau_{\text{int},\mathcal{O}}$ from our simulation data using the standard estimators discussed in~\cite{SokalLectures}.
By far the slowest of these observables is $\mathcal{N}_l$.
In Fig.~\ref{fig_corr_nl} we plot $\rho_{\mathcal{N}_l} (t/\tau_{\text{int},\mathcal{N}_l})$ and observe that
the decay is very close to being a pure exponential, suggesting $z_{\rm exp} \approx z_{{\rm int},\mathcal{N}_l}$.
\begin{figure}
\begin{center}
\includegraphics[scale=0.8]{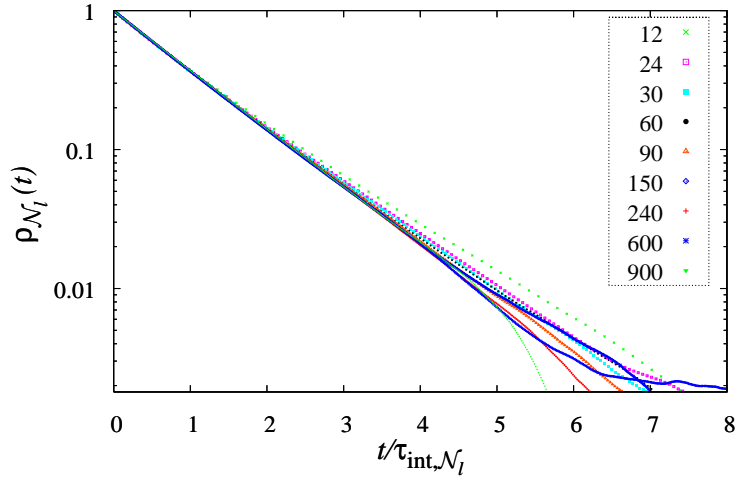}
\caption{\label{fig_corr_nl}Autocorrelation function $\rho_{\mathcal{N}_l} (t)$ versus 
$t/\tau_{{\rm int},\mathcal{N}_l}$. It is clear that  $\rho_{\mathcal{N}_l} (t)$  decays
almost as a pure exponential, suggesting $z_{\rm exp} \approx z_{{\rm int},\mathcal{N}_l}$.}
\end{center}
\end{figure}
We fitted the $\tau_{{\rm int},\mathcal{N}_l}$ data to the ansatz
\begin{equation}
\tau_{{\rm int}} = a + b L^{z_{{\rm int}} + 2X_{h}},
\label{fit_tau_int}
\end{equation}
which produced the estimate
\begin{equation*}
z_{{\rm int},\mathcal{N}_l}= 0.515 (8),
\end{equation*}
suggesting
\begin{equation*}
z_{\exp}= 0.515 (8).
\end{equation*}

Assuming $z_{\exp}=z_{\text{int},\mathcal{N}_l}$, 
the long-time decay of the autocorrelation function for any observable $\mathcal{O}$ should behave like
$\rho_{\mathcal{O}}(t)\sim \exp(-t/\tau_{\text{int},\mathcal{N}_l})$.
However, it was observed in~\cite{DengGaroniSokal07c} that for the standard worm algorithm simulating the critical Ising model on the square and
simple cubic lattices, some observables can have quite unusual {\em short-time} dynamics. Indeed it was found that $\mathcal{D}_0$ decorrelated
in $O(1)$ hits and a detailed investigation of $\rho_{\mathcal{D}_0}(t)$ was presented. 
This phenomenon in which some observables decorrelate on time scales much less than $L^{z_{\exp}}$
has been dubbed {\em critical speeding-up}~\cite{DengGaroniSokal07a,DengGaroniSokal07b,DengGaroniSokal07c}.
We have not performed a detailed investigation of the behavior of $\rho_{\mathcal{D}_0}(t)$ here, however we note that
$\tau_{\text{int},\mathcal{D}_0}\approx 0.5$, independent of $L$, showing clearly that $\mathcal{D}_0$ certainly exhibits critical speeding-up under the dynamics of
Algorithm~\ref{fpl algorithm}.
For $\mathcal{L}_2$ and $\mathcal{G}_2$ the short-time decay of $\rho(t)$ appears to be intermediate between that of $\mathcal{N}_l$ and $\mathcal{D}_0$.
To illustrate, in Fig.~\ref{fig_corr_l2} we plot $\rho_{\mathcal{L}_2} (t/\tau_{\text{int},\mathcal{N}_l})$.
\begin{figure}
\begin{center}
\includegraphics[scale=0.8]{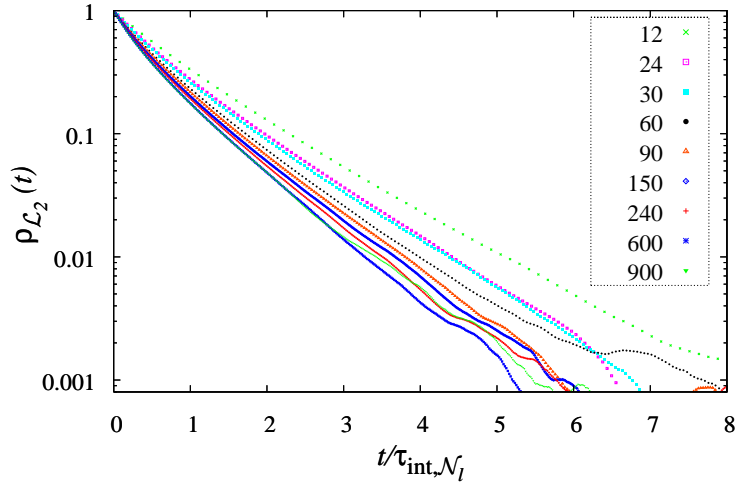}
\caption{\label{fig_corr_l2}Autocorrelation function $\rho_{\mathcal{L}_2} (t)$ versus $t/\tau_{{\rm int},\mathcal{N}_l}$.}
\end{center}
\end{figure}
It appears that $\rho_{\mathcal{L}}(t)$ has a short-time decay on a time scale strictly less than $L^{z_{\exp}}$.
Similar behavior is observed for $\rho_{\mathcal{G}_2}(t)$.

\section{Discussion} 
\label{discussion}
We have formulated a worm algorithm that correctly
simulates the FPL model on the honeycomb lattice when $n=1$.
Furthermore, we have rigorously proved its validity by showing that
the corresponding Markov chain is irreducible and has uniform
stationary distribution. Using standard duality relations this algorithm can also be used to simulate
the zero-temperature triangular-lattice antiferromagnetic Ising model.

We have tested this worm algorithm
numerically and estimate $z_{\exp}=0.515(8)$, which suggests that
it suffers from only mild critical slowing
down. We observe that the dynamics of the algorithm exhibits the
multi-time-scale behavior observed in~\cite{DengGaroniSokal07c}. It
would be interesting to to examine the dynamic behavior of observables other than $\mathcal{N}_l$ in more detail,
along the lines presented in~\cite{DengGaroniSokal07c}, but this we
leave to future work. We also obtained some interesting results regarding the static behavior of the FPL model, notably that
the face-size moments appear to be governed by the magnetic dimension for percolation. 
This is consistent with the argument in~\cite{BloteNienhuis94} that the FPL model for general $n$ displays simultaneously
the universal properties of a densely-packed loop model and those of a model with central charge $c=1$ and thermal dimension $X_t=1$.

Finally, we note that one could in principle simulate (\ref{eulerian-subgraph measure}) with $n>1$ by
incorporating appropriate connectivity checking into the Metropolis acceptance probabilities,
or by combining an $n=1$ worm algorithm with a ``Chayes-Machta coloring'' as described in \cite{DengGaroniGuoBloteSokal07}.

\begin{ack}
This research was supported in part by the Alexander von Humboldt
Foundation, and by NSF grant PHY-0424082. TMG is grateful for the
support of the Australian Research Council through the ARC Centre of
Excellence for Mathematics and Statistics of Complex Systems. Y.D. acknowledges the support of the Science Foundation of The Chinese Academy of Sciences.
YD and TMG are indebted to Alan Sokal and Wenan Guo for helpful discussions. TMG is grateful for the hospitality shown by the
University of Science and Technology of China at which this work was completed, and particularly grateful to Prof Bing-Hong Wang, as well as
the Hefei National Laboratory for Physical Sciences at Microscale.
\end{ack}

\appendix
\section{Topological constraints on fully-packed subgraphs of the honeycomb lattice}
\label{honeycomb topology}
The following lemmas describe some topological constraints on fully-packed spanning subgraphs of the honeycomb lattice with periodic boundary conditions. They
are completely independent of any considerations regarding worm algorithms. We make essential use of Lemma~\ref{neighboring defects lemma} in the proof 
of Proposition~\ref{closed proposition}.
\begin{lemma}
\label{neighboring defects lemma}
If $(A,u,v)\in \mathcal{R}$ and $u\sim v$ then $d_u(A)=d_v(A)$.
\end{lemma}

\begin{lemma}
\label{edge conservation}
Let $(A,u,v)\in\mathcal{R}$, and let $H_i$, $U_i$, $D_i$ denote, respectively, the number of vacant horizontal edges, 
the number of occupied up-pointing vertical edges, and the number of occupied down-pointing vertical edges, in row $i$. 
If row $i$ contains no defects we have $H_i=U_i=D_i$.
\begin{proof}
Fix a configuration $(A,u,v)\in\mathcal{R}$, and a row $i$ which contains no defects.
Full-packing then implies that every vertex in this row has degree 2, 
so that for every vacant horizontal edge, one of its endpoints must be adjacent to an occupied up-pointing vertical edge and its other endpoint must be adjacent
to an occupied down-pointing vertical edge, so $H_i\le U_i$ and $H_i\le D_i$. See Fig~\ref{defect free}.
Conversely, if $u$ is a vertex in row $i$ which is adjacent to an occupied down-pointing vertical edge then precisely one of its horizontal edges must be vacant,
so $D_i\le H_i$ and therefore
$D_i=H_i$.
Similarly, if $v$ is a vertex in row $i$ which is adjacent to an occupied up-pointing vertical edge then precisely one of its horizontal edges must be vacant,
so $U_i\le H_i$ and therefore
$U_i=H_i$. 
\end{proof}
\end{lemma}

\begin{figure}
\caption{\label{defect free}
  If the horizontal edge $uv$ is vacant, and neither $u$ nor $v$ is a defect, the remaining edges incident to both $u$ and $v$ are forced to be occupied,
  in a fully packed configuration.
  Conversely, if the up-pointing vertical edge $vw$ is occupied, then precisely one of the horizontal edges incident to $v$ must be vacant (here $uv$),
  and the down-pointing vertical edge $ut$ must then be occupied.
  In the diagram, thick edges are occupied, thin edges are vacant, and dotted edges are unconstrained by the state of $uv$. 
}
\begin{center}
  \begin{pspicture}(-3,-2)(2,2)
    {\psset{linewidth=0.03,linestyle=dotted}
      \pcline(-3,0)(-2,0)\pcline(1,0)(2,0)
      \pcline(1,0)(1,-1)\pcline(-2,0)(-2,1)
    }
    \pcline[linewidth=0.02](-1,0)(0,0)
    {\psset{linewidth=0.1}
      \pcline(0,0)(1,0)\pcline(0,0)(0,1)\pcline(-2,0)(-1,0)\pcline(-1,0)(-1,-1)
    }
    \psdots*[dotscale=1.8](-2,0)(-1,0)(0,0)(1,0)(-1,-1)(0,1)
    \uput[90](-1,0){$u$}\uput[270](0,0){$v$}\uput[90](0,1){$w$}\uput[270](-1,-1){$t$}
  \end{pspicture}
\end{center}
\end{figure}

\begin{proof}[Proof of Lemma~\ref{neighboring defects lemma}]
Let us first note that there do indeed exist configurations with $d_u(A)=d_v(A)$ when $(A,u,v)\in\mathcal{R}$ and $u\sim v$. Indeed, if $(A,u,u)\in\mathcal{R}$ and 
$v\sim u$
then $(A\triangle uv,u,v)\in\mathcal{R}$; 
if $uv\in A$ then $d_u(A\triangle uv)=1=d_v(A\triangle uv)$, whereas if $uv \not\in A$ then $d_u(A\triangle uv) = 3 = d_v(A\triangle uv)$.

Suppose on the contrary that $(A,u,v)\in\mathcal{R}$ with $u\sim v$, but $d_u(A)\neq d_v(A)$.
Since $u\neq v$ we must have $d_u(A), d_v(A)\in \{1,3\}$, so that either $d_u(A)=3$ and $d_v(A)=1$, or vice versa.
Let us assume (without loss of generality) the former.
There are two possibilities for the edge $uv$; either $uv$ is a horizontal edge, so that $u$ and $v$ lie in the same row, or $uv$ is a vertical edge, so
that $u$ and $v$ lie in adjacent rows. 

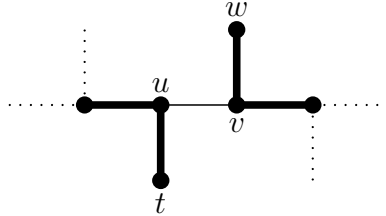
\begin{figure}
  \caption{\label{uv in same row}Neighborhood of the horizontal edge $uv$ with $d_u(A)=3$ and $d_v(A)=1$.}
  \begin{center}
    \begin{pspicture}(-2,-2)(2,2)
      \pcline[linewidth=0.02](0,0)(0,-1)\pcline[linewidth=0.02](0,0)(1,0)
             {\psset{linewidth=0.1}
               \pcline(-2,0)(-1,0)\pcline(-1,0)(0,0)\pcline(-1,0)(-1,1)
               \pcline(1,0)(2,0)\pcline(1,0)(1,1)
             }
             \psdots*[dotscale=1.8](-2,0)(-1,0)(0,0)(1,0)(-1,1)(0,-1)(1,1)(2,0)
             \uput[270](-1,0){$u$}\uput[90](0,0){$v$}\uput[270](1,0){$w$}\uput[90](-1,1){$u'$}\uput[270](0,-1){$v'$}\uput[90](1,1){$w'$}
    \end{pspicture}
  \end{center}
\end{figure}

Suppose $uv$ is a horizontal edge lying in row $i$, denote $v$'s other horizontal edge by $vw$,
and suppose that the vertical edge $uu'$ is up-pointing, so that the vertical edge $vv'$ must be down-pointing.
See Fig.~\ref{uv in same row}. The up-pointing vertical edges $uu'$ and $ww'$ are both
occupied. Suppose there are $n$ other occupied up-pointing vertical
edges incident to row $i$, so there are $n+2$ in total. Each of these
other $n$ occupied up-pointing vertical edges is incident to a degree 2 vertex
in row $i$. Let $a$ be such a vertex, then $a$ must have one of its horizontal
edges vacant, call it $ab$. By assumption we have $a\neq u,w$, so that
$b\neq v$, and so $d_b(A)=2$ and $b$ must have its vertical edge (which is down-pointing)
occupied. Therefore, every one of the $n$ occupied up-pointing vertical
edges other than $uu'$ and $ww'$ corresponds to an occupied
down-pointing vertical edge. Conversely, if there is an
occupied down-pointing vertical edge incident to some $b\neq v$ in row $i$
then $b$ must have one of its horizontal edges vacant, call it $ab$.
Since $b\neq v$ and $ab$ is vacant we have $a\neq u,w$, so that $d_a(A)=2$. 
Therefore $a$ must have its vertical edge (which is up-pointing) occupied, and this edge is
neither $uu'$ nor $ww'$. Therefore there are  $n$ occupied
down-pointing and $n+2$ occupied up-pointing vertical edges incident
to row $i$. Now, since no other row contains a defect,
Lemma~\ref{edge conservation} tells us that all rows below row $i$
will have $n$ occupied up-pointing and down-pointing vertical edges,
whereas all rows above row $i$ will have  $n+2$ occupied up-pointing
and down-pointing vertical edges. However, it is impossible for this
to occur if we have periodic boundary conditions, and so we have a
contradiction.  Of course, if we assume instead that $u$ is
down-pointing and $v$ up-pointing then an entirely similar argument
leads to a similar contradiction. Therefore if $(A,u,v)\in\mathcal{R}$
and $uv$ is a horizontal edge we must have $d_u(A)=d_v(A)$.

The converse situation where $uv$ is a vertical edge can be treated in a similar manner. We omit the details.
\end{proof}

\section{Lemmas used in the proof of Proposition~\ref{irreducibility-transience proposition}}
\label{irreducibility lemmas appendix}
\begin{lemma}
\label{meta-lemma}
Let $(A,v,v)\in\mathcal{S}$ with $vv'\not\in A$. Then whenever $d_v(A)=2$ we have
$$
(A,v,v)\rightarrow(A\cup vv', v',v)\rightarrow(A,v',v'),
$$
and whenever $d_v(A)=d_{v'}(A)=2$ we have
$$
(A,v,v)\rightarrow(A\setminus vv'',v'',v)\rightarrow(A,v'',v''),
$$
for both $v''\sim v$ with $v''\neq v'$.
\end{lemma}

\begin{lemma}
\label{transient lemma}
  Let $H\in\fp$ denote the fully-packed configuration in which every horizontal edge is occupied and every vertical edge is vacant.
  Suppose $(A,u,u)\in\mathcal{S}$ with $d_u(A)=2$ and $A\neq H$.
  Then there always exists a vacant horizontal edge $vv'$ for which $(A,u,u)\rightarrow(A,v,v)$.
\end{lemma}

\begin{lemma}
\label{eulerian intercommunication}
If $(A,v,v)\in\mathcal{R}$ then $(A,v,v)\leftrightarrow(A,u,u)$ for any pair of vertices $u$ and $v$.
\end{lemma}

\begin{lemma}
\label{quasi-eulerian irreducibility}
If $(A,v,v)\in\mathcal{R}$ and $u\sim v$ then 
$$(A,v,v)\leftrightarrow(A\triangle uv,u,v).$$
\end{lemma}

\begin{lemma}
\label{degree 1 intercommunication}
Let $(A,u,v)\in\mathcal{R}$ with $d_u(A)=1$ and suppose $uu'\not\in A$. Then
$$
(A,u,v)\leftrightarrow (A\cup uu',u',v)
$$
\end{lemma}

\begin{lemma}
\label{degree 3 intercommunication}
If $(A,u,v)\in\mathcal{R}$ with $d_u(A)=3$ then for each $u'\sim u$
$$
(A,u,v)\leftrightarrow(A\setminus uu',u',v).
$$
\end{lemma}

\begin{proof}[Proof of Lemma~\ref{meta-lemma}]
For any $(A,v,v)\in\mathcal{S}$ with $vv'\not\in A$ we have
$$
\pworminfty[(A,v,v)\rightarrow(A\cup vv',v',v)]=1/6.
$$
Furthermore, if $d_v(A)=2$ then $d_v(A\cup vv')=3$ so
$$
\pworminfty[(A\cup vv',v',v)\rightarrow(A,v',v')]=1/6
$$
and we have 
$$
(A,v,v)\rightarrow(A\cup vv',v',v)\rightarrow(A,v',v').
$$

If in fact $d_v(A)=2=d_{v'}(A)$ then for both $v''\sim v$ with $v''\neq v'$ we have
\begin{align*}
\pworminfty[(A,v,v)\rightarrow(A\cup vv',v,v')] &= 1/6\\
\pworminfty[(A\cup vv', v,v')\rightarrow(A\cup vv'\setminus vv'',v'',v')] &= 1/6\\
\pworminfty[(A\cup vv'\setminus vv'',v'',v')\rightarrow(A\setminus vv'',v'',v)] &=1/6\\
\pworminfty[(A\setminus vv'',v'',v)\rightarrow(A,v'',v'')] &= 1/4\\
\end{align*}
so that
$$
(A,v,v)\rightarrow(A\setminus vv'',v'',v)\rightarrow(A,v'',v'').
$$
\end{proof}

\begin{proof}[Proof of Lemma~\ref{transient lemma}]
  Denote by $uu'\not\in A$ the unique vacant edge incident to $u$. There are two possibilities: either $d_{u'}(A)=0$ or $d_{u'}(A)=2$.
  If $d_{u'}(A)=0$ then $u'$ has both its incident horizontal edges vacant, and since Lemma~\ref{meta-lemma} implies
  $(A,u,u)\rightarrow(A,u',u')$ there is nothing more to show. If on the other hand $d_{u'}(A)=2$ then Lemma~\ref{meta-lemma}
  implies $(A,u,u)\rightarrow(A,v,v)$ for every $v\sim u$. If any of the $(A,v,v)$ have a vacant horizontal edge incident to $v$ we are done.
  Otherwise we re-apply Lemma~\ref{meta-lemma} to $(A,v,v)$ for every $v\sim u$. In this way we must eventually arrive at some $(A,w,w)$ 
  for which there is a vacant horizontal edge incident to $w$. Transitivity implies $(A,u,u)\rightarrow(A,w,w)$ and the stated result follows.
\end{proof}

\begin{proof}[Proof of Lemma~\ref{eulerian intercommunication}]
If $(A,v,v)\in \mathcal{R}$ then in fact $(A,u,u)\in\mathcal{R}$ for any $u$. Since every vertex has degree 2 we can apply Lemma~\ref{meta-lemma} to
$(A,v,v)$ to see that $(A,v,v)\rightarrow(A,v',v')$ for any $v'\sim v$, but we can equally apply it to $(A,v',v'$) to see that $(A,v',v')\to(A,v,v)$.
So for any pair of neighboring vertices $v\sim v'$ we have $(A,v,v)\leftrightarrow(A,v',v')$. Since the lattice is connected and every vertex has degree 2 this 
immediately extends, via transitivity of $\leftrightarrow$, to $(A,v,v)\leftrightarrow(A,u,u)$ for any arbitrary pair or vertices $u,v$.
\end{proof}

\begin{proof}[Proof of Lemma~\ref{quasi-eulerian irreducibility}]
Let $(A,v,v)\in\mathcal{R}$ and $u\sim v$. Lemma~\ref{meta-lemma} immediately implies
$$
(A,v,v)\rightarrow(A\triangle uv,u,v)
$$
and combining Lemma~\ref{meta-lemma} with Lemma~\ref{eulerian intercommunication} we further obtain
$$
(A\triangle uv,u,v)\rightarrow(A,u,u)\rightarrow(A,v,v).
$$
Therefore $(A,v,v)\leftrightarrow(A\triangle uv,u,v)$.
\end{proof}

\begin{proof}[Proof of Lemma~\ref{degree 1 intercommunication}]
Suppose that $u'\neq v$. Then
$$
\pworminfty[(A,u,v)\to(A\cup uu',u',v)] = \frac{1}{4},
$$
and since $u'\neq v$ implies $d_{u'}(A)=2$ we have $d_{u'}(A\cup uu')=3$, so that
\begin{equation}
\pworminfty[(A\cup uu',u',v)\to(A,u,v)] = \frac{1}{6}.
\end{equation}
Therefore $(A,u,v)\leftrightarrow(A\cup uu',u',v)$ when $u'\neq v$.

Conversely, suppose $u'=v$. Then Lemma~\ref{quasi-eulerian irreducibility} implies that
$(A\cup uv, v,v)\leftrightarrow(A,u,v)$.
\end{proof}

\begin{proof}[Proof of Lemma~\ref{degree 3 intercommunication}]
Suppose that $u'\neq v$. Since $d_u(A)=3$
$$
\pworminfty[(A,u,v)\to(A\setminus uu',u',v)] = \frac{1}{6}.
$$
Furthermore, since $u'\neq v$ we have $d_{u'}(A)=2$ and hence $d_{u'}(A\setminus uu')=1$, so
$$
\pworminfty[(A\setminus uu',u',v)\to(A,u,v)] = \frac{1}{4}.
$$
Therefore $(A,u,v)\leftrightarrow(A\setminus uu',u',v)$ when $u'\neq v$.

Conversely, suppose $u'=v$. Then Lemma~\ref{quasi-eulerian irreducibility} implies that
$(A\setminus uv, v,v)\leftrightarrow(A,u,v)$.
\end{proof}

\bibliographystyle{elsart-num}

\end{document}